\documentclass[11pt]{article}

\usepackage[square,numbers,sort&compress]{natbib}

\usepackage{tablefootnote}
\usepackage[utf8]{inputenc} 
\usepackage[T1]{fontenc}    
\usepackage{hyperref}       
\usepackage{url}            
\usepackage{booktabs}       
\usepackage{amsfonts}       
\usepackage{nicefrac}       
\usepackage{microtype}      
\usepackage{tikz}
\usetikzlibrary{decorations.pathreplacing,calc}

\usepackage[algo2e,ruled]{algorithm2e}
\LinesNumbered

\usepackage{verbatim}
\usepackage{amsthm}
\usepackage{amsfonts}
\usepackage[nonamelimits]{amsmath}
\usepackage{amssymb}
\usepackage{sansmath}
\usepackage{color}
\usepackage{xcolor}
\usepackage{bm}
\usepackage{bbm}
\usepackage{fullpage}
\usepackage[bottom]{footmisc}

\theoremstyle{plain}
\newtheorem{theorem}{Theorem}[section]

\newtheorem{lemma}[theorem]{Lemma}
\newtheorem{fact}[theorem]{Fact}
\newtheorem{proposition}[theorem]{Proposition}
\newtheorem{remark}[theorem]{Remark}

\theoremstyle{definition}
\newtheorem{definition}[theorem]{Definition}

\usepackage{subfig}
\usepackage{graphicx}
\usepackage{listings}
\usepackage{enumitem}
\usepackage{tabularx}
\newtheoremstyle{named}{}{}{\itshape}{}{\bfseries}{.}{.5em}{\thmnote{#3}}
\theoremstyle{named}

\usepackage{tikz}
\usetikzlibrary{fit,calc}

\usepackage{tikz}
\usetikzlibrary{cd}

\DeclareSymbolFont{extraup}{U}{zavm}{m}{n}
\DeclareMathSymbol{\varheart}{\mathalpha}{extraup}{86}
\DeclareMathSymbol{\vardiamond}{\mathalpha}{extraup}{87}

\usepackage{mathtools}

\makeatletter
\newcommand*{\inlineequation}[2][]{%
  \begingroup
    \refstepcounter{equation}%
    \ifx\\#1\\%
    \else
      \label{#1}%
    \fi
    \relpenalty=10000 %
    \binoppenalty=10000 %
    \ensuremath{%
      #2%
    }%
    ~\@eqnnum
  \endgroup
}
\makeatother
\usepackage[table]{xcolor}

\usepackage{threeparttable}

\usepackage{todonotes}

\usepackage{xspace}

\xspaceaddexceptions{[]\{\}}

\usepackage{hyperref}

\usepackage{thm-restate}

\usepackage{bm}
\makeatletter
\newcommand*{\rom}[1]{\expandafter\@slowromancap\romannumeral #1@}

\usepackage{multirow}

\usepackage{tablefootnote}
\usepackage{float}
\usepackage{makecell}
\usepackage{authblk}

\allowdisplaybreaks

\title{On Complexity of Gaussian Sampling}
\date{\vspace{-5ex}}
\begin{document}
\author[1]{Huanjian Zhou\thanks{zhou-huanjian185@g.ecc.u-tokyo.ac.jp}}
\author[1,2]{Taiji Suzuki\thanks{taiji@mist.i.u-tokyo.ac.jp }}
\affil[1]{The University of Tokyo}
\affil[2]{RIKEN Center for Advanced Intelligence Project}

\maketitle

\begin{abstract}%
We study the number of exact matrix-vector queries needed to sample a centered $d$-dimensional Gaussian with condition number at most $\kappa$ to total variation error $\varepsilon$.
The best known algorithm uses $O(\min\{d,\sqrt{\kappa}\log(\frac{d}{\varepsilon})\})$ queries, while
the known lower bounds do not grow as $\varepsilon\to0$ and match this bound only for $\kappa\le
d^{1/5-\delta}$ at constant accuracy.
We introduce a construction based on a hidden determinant sign and prove a lower bound of $\Omega\!\left(\min\left\{d,\sqrt{\kappa}\log\left(1+\frac{d}{\sqrt{\kappa}\,\varepsilon}\right)\right\}\right)$ uniformly for all $d\ge1$, $\kappa\ge4$, and $0<\varepsilon\le\tfrac{1}{16}$.
This bound matches the known upper bound up to constant factors whenever $\kappa\le d^{2-\delta}$ for fixed $\delta\in(0,2)$, or $\varepsilon\le d^{-c}$ for fixed $c>0$.
We also extend a recent fixed-spectrum construction with a hidden eigenbasis from $\kappa=2$ to
$2\le\kappa\le d^{1-\delta}$, for any fixed $\delta\in(0,1)$. This gives an
$\Omega_\delta(\sqrt{\kappa}\log d)$ lower bound at constant accuracy that holds even when the
spectrum of the precision matrix is known to the algorithm.
\end{abstract}

\newpage
\setcounter{tocdepth}{2}
\tableofcontents

\newpage

\section{Introduction}
We study the problem of sampling from a centered Gaussian distribution given query access to its
unnormalized density.
Gaussian sampling is a fundamental algorithmic task
with applications in Bayesian inference~\citep{bhattacharya2016fast},
numerical simulation~\citep{chow2014preconditioned},
and randomized approximation algorithms~\citep{goemans1995improved}.

The target is $\pi_A=N(0,A^{-1})$, where the unknown precision
matrix $A$ satisfies $I_d\preceq A\preceq\kappa I_d$.
An algorithm accesses $A$ through matrix--vector product queries.
At each step, it chooses a vector $x\in\mathbb R^d$ and observes $Ax$. 
The goal of the algorithm is to output a random vector whose distribution
is within $\varepsilon$ of $\pi_A$ in total variation (TV) distance.
We ask how many queries required for sampling from \(d\)-dimensional Gaussian distribution to accuracy $\varepsilon$ with the condition number $\kappa$.

Gaussian sampling has a long history, with methods based on polynomial
approximations to matrix square roots and inverse square
roots~\citep{fixman1986construction,davis1987generating},
Krylov subspaces and conjugate
gradients~\citep{parker2012sampling,aune2013iterative,chow2014preconditioned},
and Markov chain Monte Carlo, including Gibbs sampling and Hamiltonian
Monte Carlo~\citep{roberts1997updating,fox2018accelerated,apers2024hamiltonian}.
%
In the first-order oracle model considered here,
\cite{lee2020logsmooth} obtained a 
query bound of
$O(\kappa d\log^3(\tfrac{\kappa d}{\varepsilon}))$
for Metropolized Hamiltonian Monte Carlo.
For Gaussian targets, \cite{apers2024hamiltonian} improved this to
$\widetilde O(\sqrt{\kappa}\,d^{1/4}\log(\tfrac{1}{\varepsilon}))$
using long, randomized integration times, with initialization included
in the query count.
Analyses of Hit-and-Run also yield an
$O(\kappa d\log(\tfrac{\kappa d}{\varepsilon}))$
query bound, since each exact Gaussian line update requires only one
matrix-vector
query~\citep{ascolani2024entropy,bou2025accelerated}.
A sharper worst-case upper bound is provided by the polynomial sampler
of \cite{chewi2024query}, which requires
$O(\min\{d,\sqrt{\kappa}\log(\tfrac{d}{\varepsilon})\})$
queries for total variation error $\varepsilon$.
At a sufficiently small constant total variation error,
they proved a lower bound of
$\Omega(\min\{d,\sqrt{\kappa}\})$
against arbitrary randomized adaptive algorithms.
They further obtained
$\Omega_\delta(\sqrt{\kappa}\log d)$
when $\kappa\le d^{1/5-\delta}$,
for any fixed $\delta\in(0,1/5)$,
matching their upper bound at constant accuracy in this regime.
They conjectured that $\Omega(\sqrt{\kappa}\log d)$ queries are necessary whenever
$\sqrt{\kappa}\log d\le d$. A recent preprint~\cite{OAI:Subpolynomial-query-complexity-for-well-conditioned-log-concave-sampling-September-26-2026} gives an $\Omega(\log d)$ lower bound for
$\kappa=2$ and $\varepsilon=\frac{1}{10}$, using instances with a fixed spectrum and a random
eigenbasis. None of these lower bounds grows as $\varepsilon\to0$, and for $\kappa\ge d^{1/5}$ they
are below the upper bound by a factor of order $\log(d/\varepsilon)$.

%
However, both regimes arise in applications.
%
%
%
%
%
Large condition numbers arise naturally in spatial statistics and Gaussian process regression~\citep{ababou1994condition,banerjee2013efficient,lim2017numerical}.
On a bounded domain, the covariance matrices of a fixed Mat\'ern model satisfy $\kappa=\Omega(d)$ as the number $d$ of sampling locations increases~\citep{stein2012difference}.
Sparse precision matrices arising from SPDE-based Gaussian Markov random field approximations of Mat'ern fields can likewise become ill-conditioned under mesh refinement~\citep{lindgren2011explicit,pereira2019efficient}.
For example, a regular-grid discretization of the Mat\'ern SPDE with $\nu=1$ on a fixed two-dimensional periodic domain satisfies $\kappa=\Theta(d^2)$ under uniform refinement.

High sampling accuracy matters for estimating a rare failure probability $p$~\citep{liu2015efficient}, since TV error $\varepsilon$ allows relative bias up to $\varepsilon/p$.
High accuracy matters in Gibbs samplers with Gaussian conditionals~\citep{bhattacharya2016fast,nishimura2023prior}, where uncorrected sampling errors can bias the chain~\citep{gilavert2014efficient}.
From a common initial law, a uniform per-update TV error of $\varepsilon$ accumulates to at most $N\varepsilon$ over $N$ updates.
%
%
%
In neither regime do the known bounds determine the query complexity.
This leads to the following question.
\begin{center}
\emph{Can we establish sharp lower bounds for \underline{high-accuracy} sampling from \underline{ill-conditioned} Gaussians, \underline{uniformly over all dimensions}?}
\end{center}

\subsection{Our results}
\label{subsec:our-results}
Our first result is the following lower bound, which holds for every $\kappa\ge 4$ and increases as $\varepsilon$ decreases to zero.
\begin{table}[ht]
\centering
\small
\setlength{\tabcolsep}{4pt}
\renewcommand{\arraystretch}{1.15}
\caption{Upper and lower query bounds for sampling from $N(0,A^{-1})$
with $I_d\preceq A\preceq\kappa I_d$, using an exact first-order
oracle and a fixed worst-case query budget, to total variation
accuracy $\varepsilon$. Here $d$ is the dimension and $\kappa\ge2$
is the condition-number bound.}
\label{tab:exact-oracle-lower-bounds}
\begin{tabular}{@{}
  >{\raggedright\arraybackslash}
    p{\dimexpr0.18\linewidth-1.5\tabcolsep\relax}
  >{\raggedright\arraybackslash}
    p{\dimexpr0.21\linewidth-1.5\tabcolsep\relax}
  >{\raggedright\arraybackslash}
    p{\dimexpr0.38\linewidth-1.5\tabcolsep\relax}
  >{\raggedright\arraybackslash}
    p{\dimexpr0.23\linewidth-1.5\tabcolsep\relax}
@{}}
\toprule
Reference & Regime & Query bound & Technique \\
\midrule
\multicolumn{4}{@{}l}{\textit{Upper bounds}} \\
\addlinespace[2pt]
\cite{chewi2024query}
  & all $d,\kappa$,\newline
    $0<\varepsilon\le\tfrac1{10}$
  & $\displaystyle O\!\left(
      \min\bigl\{d,
      \sqrt{\kappa}\log(\tfrac{d}{\varepsilon})\bigr\}
      \right)$
  & approximation of $A^{-1/2}$ \\
\midrule
\multicolumn{4}{@{}l}{\textit{Lower bounds}} \\
\addlinespace[2pt]
\cite{chewi2024query}
  & all $d,\kappa$,\newline
    $0<\varepsilon\le c$
  & $\Omega(\min\{d,\sqrt{\kappa}\})$
  & hidden minimum eigenvalue \\
  \addlinespace
\cite{chewi2024query}
  & $\kappa\le d^{1/5-\delta}$,\newline
    $0<\varepsilon\le\varepsilon_0$
  & $\Omega_\delta(\sqrt{\kappa}\log d)$
  & hidden eigenvalue multiplicities \\
  \addlinespace
\cite{OAI:Subpolynomial-query-complexity-for-well-conditioned-log-concave-sampling-September-26-2026}
  & $d\geq d_0,\kappa=2$,\newline
    $\varepsilon=\tfrac1{10}$
  & $\Omega(\log d)$
  & hidden eigenbasis \\
\addlinespace
Theorem~\ref{thm:gaussian-lower-bound}
  & $d\ge1$, $\kappa\ge4$,\newline
    $0<\varepsilon\le\tfrac1{16}$
  & $\displaystyle
      \begin{aligned}[t]
      &\Omega\big(\min\big\{d,\sqrt{\kappa}\log\big(1+
      \tfrac{d}{\sqrt{\kappa}\,\varepsilon}\big)
      \big\}\big)
      \end{aligned}$
  & hidden determinant sign \\
\addlinespace
Theorem~\ref{thm:openai-general-kappa}
  &  $d\geq d_0(\delta)$,\newline
    $2\le\kappa\le d^{1-\delta}$,\newline
    $0<\varepsilon\le\tfrac1{10}$
  & $\Omega_\delta(\sqrt{\kappa}\log d)$
  & hidden eigenbasis\newline
    (extended from~\cite{OAI:Subpolynomial-query-complexity-for-well-conditioned-log-concave-sampling-September-26-2026}) \\
\bottomrule
\end{tabular}

\vspace{3pt}
\begin{minipage}{\linewidth}
\footnotesize
Here $c,\varepsilon_0>0$ are sufficiently small universal constants. $d_0(\delta)$ is a sufficiently large constant which depends on $\delta$.
In the second lower-bound row, $\delta\in(0,1/5)$ is fixed;
in the last row, $\delta\in(0,1)$ is fixed.
Constants implicit in $\Omega_\delta$, and the dimension threshold
in the last row, may depend on $\delta$.
\end{minipage}
\end{table}
\begin{theorem}[Informal; see Theorem~\ref{thm:gaussian-lower-bound}]
\label{thm:informal-gaussian-lower-bound}
For every integer
$d\ge1$, $\kappa\ge4$, and $0<\varepsilon\le1/16$, any Gaussian sampler 
requires $\Omega\big(\min\big\{
        d,\,
        \sqrt{\kappa}
        \log\big(1+\frac{d}{\sqrt{\kappa}\,\varepsilon}\big)
    \big\}\big)$ queries.
\end{theorem}

Our lower bound matches the known upper bound of $ O\!\big(
      \min\big\{d,
      \sqrt{\kappa}\log(\tfrac{d}{\varepsilon})\big\}
      \big)$ up to constant factors whenever $\kappa\le d^{2-\delta}$ for any fixed $\delta\in(0,2)$. 
      At constant accuracy, the two bounds were previously
known to match only for $\kappa\le d^{1/5-\delta}$.
This confirms the conjecture of  \cite{chewi2024query} that $\Omega(\sqrt{\kappa}\log d)$ queries are necessary at constant accuracy whenever $\sqrt{\kappa}\log d\le d$.
It also matches the upper bound for every $\kappa\ge4$ whenever $\varepsilon\le d^{-c}$ for any fixed $c>0$.
For $\kappa\ge d^2$, it yields the optimal $\Theta(d)$ query complexity.

%

We also extend the hidden-eigenbasis construction of \cite{OAI:Subpolynomial-query-complexity-for-well-conditioned-log-concave-sampling-September-26-2026} to obtain an $\Omega_\delta(\sqrt{\kappa}\log d)$ lower bound for $2\le\kappa\le d^{1-\delta}$. We summarize the comparisons in Table~\ref{tab:exact-oracle-lower-bounds}.

\begin{theorem}[Informal; see Theorem~\ref{thm:openai-general-kappa}]
For every fixed $\delta\in(0,1)$, sufficiently large $d$,
$2\le\kappa\le d^{1-\delta}$, and $0<\varepsilon\le1/10$, any Gaussian sampler requires $\Omega_\delta (\sqrt{\kappa}\log d)$ queries.
\end{theorem}
This bound is tight at constant accuracy and, more generally, whenever $d^{-C}\le\varepsilon\le1/10$ for a fixed $C>0$, with matching constants depending on $\delta$ and $C$. 
For $\kappa\ge4$ and $\varepsilon\le\frac{1}{16}$ it is implied by Theorem~\ref{thm:gaussian-lower-bound} while the new
cases are $2\le\kappa<4$ and $\frac{1}{16}<\varepsilon\le\frac{1}{10}$. We include it because its
hard instances have a fixed spectrum, so the bound holds even for algorithms that know the spectrum
of the precision matrix.
%
Section~\ref{subsec:openai-limitations} explains why the analysis of these instances gives no dependence on
$\varepsilon$ and does not extend beyond $\kappa = d^{1-\delta}$.

\begin{remark}[\bf Connection to strongly log-concave sampling]
Our lower bounds also apply to sampling from $\pi\propto e^{-V}$
with $I_d\preceq\nabla^2V\preceq\kappa I_d$ using exact
value-and-gradient queries, even when the minimizer is known.
For $\kappa=2$ and TV accuracy $1/10$,
\cite{OAI:Subpolynomial-query-complexity-for-well-conditioned-log-concave-sampling-September-26-2026}
establishes an $O_\eta(d^\eta)$ upper bound for every fixed $\eta>0$,
leaving a gap from the $\Omega(\log d)$ lower bound, whereas
Gaussian sampling has complexity $\Theta(\log d)$ in this regime.
For general condition numbers and high accuracy,
\cite{chen2026accelerated} obtain TV error $\varepsilon$ using
$\widetilde O(\sqrt{\kappa}\,d^{1/4}
\operatorname{polylog}(1/\varepsilon))$
expected gradient queries with a warm start.
Thus, our results provide an improved lower bounds for general strongly log-concave sampling with almost tight lower bounds for Gaussian subclass.
\end{remark}

\section{Preliminaries}
\label{sec:preliminaries}

We write $\|\cdot\|$ for the Euclidean norm and $I_d$ for the
$d\times d$ identity matrix.
For symmetric matrices $A$ and $B$, the notation $A\preceq B$
means that $B-A$ is positive semidefinite.

Let
$O(n):=\{O\in\mathbb R^{n\times n}:O^\top O=I_n\}$
denote the real orthogonal group, equipped with Haar probability
measure.
It has two components, $\{\det O = +1\}$,
$\{\det O=-1\}$, and each have Haar probability $1/2$.
For $b\in\{-1,+1\}$, we write $\mu_{n,b}$ for Haar probability
measure conditioned on $\det O=b$.


\paragraph{Targets and oracle.}
We restrict the target class to centered Gaussians. 
This restriction removes the hardness of locating the target and isolates the difficulty of sampling with an unknown covariance.
For an integer $d\ge 1$ and a parameter $\kappa\ge 1$, without loss of generality, we normalize the spectral bounds to $1$
and $\kappa$ by rescaling and consider
the precision class
\[
    \mathcal A_{d,\kappa}
    :=
    \left\{
        A\in\mathbb R^{d\times d}:
        A=A^\top,\quad I_d\preceq A\preceq \kappa I_d
    \right\}.
\]
An input $A\in\mathcal A_{d,\kappa}$ specifies the quadratic
potential $f_A(x):=\frac12 x^\top Ax$ and the target distribution
$\pi_A:=N(0,A^{-1})$.
At a point $x\in\mathbb R^d$, the exact first-order oracle returns
\begin{equation}
    \operatorname{FO}_A(x)
    :=
    \bigl(f_A(x),\nabla f_A(x)\bigr)
    =
    \left(\frac12 x^\top Ax,Ax\right).
    \label{eq:first-order-oracle}
\end{equation}

On this class, exact first-order access is equivalent to exact
matrix--vector access.
Indeed, a first-order response contains $Ax$, while a
matrix--vector response determines both the gradient $Ax$
and the function value $\frac12\langle x,Ax\rangle$.
Thus the exact value-and-gradient, gradient-only, and
matrix--vector query complexities coincide.

\paragraph{Randomized algorithms.}
A randomized algorithm chooses its queries and stopping decisions
as Borel functions of its internal randomness and previous oracle
responses, where the randomness is independent of the input.
Its output may be any $\mathbb R^d$-valued Borel function of
the same information.
The algorithm has \emph{fixed query budget} $T$ if it makes
at most $T$ queries for every admissible input and every
realization of its internal randomness.

\paragraph{Accuracy measures.}
For probability measures $P$ and $Q$ on a common measurable space,
their total variation distance is
\[
    \operatorname{TV}(P,Q)
    :=
    \sup_B |P(B)-Q(B)|
    =
    \frac12\int |p-q|\,d\lambda,
\]
where the supremum ranges over measurable sets, and $p$ and $q$
are the densities of $P$ and $Q$ with respect to any common
dominating measure $\lambda$.
In particular,
$|\mathbb E_P h-\mathbb E_Q h|
\le \operatorname{TV}(P,Q)$
for every measurable function $h$ taking values in $[0,1]$.
\paragraph{Sampling accuracy and query complexity.}
For a fixed input $A$, let $\rho_A^{\mathrm{Alg}}$ denote
the law of the algorithm's output, with all internal
randomness included.
For $0<\varepsilon<1$, define
\begin{equation}
    Q_G(d,\kappa,\varepsilon)
    :=
    \inf\left\{
        T\in\mathbb Z_{\ge 0}:
        \begin{array}{l}
            \exists\,\mathrm{Alg}
            \text{ with fixed query budget }T
            \\[2pt]
            \text{such that }\displaystyle
            \sup_{A\in\mathcal A_{d,\kappa}}
            \operatorname{TV}
            \bigl(\rho_A^{\mathrm{Alg}},\pi_A\bigr)
            \le \varepsilon
        \end{array}
    \right\}.
    \label{eq:gaussian-query-complexity}
\end{equation}
The accuracy requirement is uniform over inputs and concerns
the output law after averaging over the algorithm's seed.
Throughout the paper, query budgets are deterministic caps
on the number of queries in each run, rather than bounds
on the expected number of queries.

\section{Lower Bound via Hidden Determinant Sign}
\label{sec:construction}
In this section, we establish a Gaussian sampling lower bound
that holds uniformly for all dimensions $d\ge1$,
arbitrarily large condition-number bounds $\kappa\ge4$,
and arbitrarily small total variation errors
$0<\varepsilon\le1/16$.
\begin{theorem}[\bf Gaussian sampling lower bound]
\label{thm:gaussian-lower-bound}
For every integer $d\ge1$, $\kappa\ge4$, and
$0<\varepsilon\le1/16$,
\begin{equation}
    Q_G(d,\kappa,\varepsilon)
    \ge
    \frac1{256}
    \min\left\{
        d,\,
        \sqrt{\kappa}
        \log\left(
            1+\frac{d}{\sqrt{\kappa}\,\varepsilon}
        \right)
    \right\}.
    \label{eq:gaussian-lower-bound}
\end{equation}
\end{theorem}

In the remainder of this section, we present the construction
in Section~\ref{subsec:hard-family}, state its three key properties
in Section~\ref{subsec:construction-properties}, and use them
to prove Theorem~\ref{thm:gaussian-lower-bound}
in Section~\ref{subsec:main-proof}.
The proofs of these properties are deferred to
Sections~\ref{sec:indistinguishability},~\ref{sec:mixture-separation} and~\ref{sec:regularity}.

\subsection{The Hard Gaussian Family}
\label{subsec:hard-family}

Fix $\kappa\ge4$.
Let $3\le n\le d$ be the block dimension and
$k:=\lfloor d/n\rfloor$ the number of blocks.
In the proof of the main theorem, we choose $n=2T+1$,
where $T$ is the algorithm's query budget.
This choice applies when $T\ge1$ and $2T+1\le d$.
The remaining cases are treated separately.
For $O\in O(n)$, define
\begin{equation}
    A_O
    :=
    \frac{\kappa+1}{2}I_n
    -
    \frac{\kappa-1}{4}(O+O^\top).
    \label{eq:single-block-precision}
\end{equation}
For $\boldsymbol O=(O_1,\ldots,O_k)\in O(n)^k$, set
\begin{equation}
    A_{\boldsymbol O}
    :=
    \operatorname{diag}
    \bigl(A_{O_1},\ldots,A_{O_k},I_{d-kn}\bigr).
    \label{eq:block-precision}
\end{equation}

For each $b\in\{-1,+1\}$, draw $O_1,\ldots,O_k$ independently
from $\mu_{n,b}$, and let $\mathcal D_b$ denote the law of
$A_{\boldsymbol O}$.
Thus the orthogonal matrices share the same determinant sign.

Writing
$P_b:=\mathbb E_{O\sim\mu_{n,b}}[N(0,A_O^{-1})]$
for the single-block target mixture, the averaged target
under $\mathcal D_b$ is
\begin{equation}
    \bar\pi_b
    :=
    \mathbb E_{A\sim\mathcal D_b}[\pi_A]
    =
    P_b^{\otimes k}\otimes N(0,I_{d-kn}).
    \label{eq:averaged-targets}
\end{equation}
When $kn=d$, the identity block and the final Gaussian factor
are omitted.

\subsection{Properties of the Construction}
\label{subsec:construction-properties}
This subsection states the three lemmas used in the proof of Theorem~\ref{thm:gaussian-lower-bound}.
Lemma~\ref{lem:construction-indistinguishability} shows that an algorithm with fewer than $n/2$
queries has the same average output law under the two priors,
Lemma~\ref{lem:construction-separation} shows that the two averaged targets are far apart in
total variation, and
Lemma~\ref{lem:construction-admissibility} verifies that the hard family belongs to
$\mathcal A_{d,\kappa}$.
All proofs are deferred to Appendix~\ref{app:b}.
\paragraph{Why the priors are indistinguishable.}
A first-order query can be simulated by observing $Ov$ and
$O^\top v$ in each block and each of these two observations fixes at most one new direction
of the unknown orthogonal map.
Conditionally on the history, the rest of $O_j$ is Haar distributed subject to its determinant sign,
and as long as two dimensions are unexplored the next response has the same law under either sign.
The condition $2T<n$ ensures this at every observation,
including when queries depend on responses from other blocks.
The proof is in Section~\ref{sec:indistinguishability}.

\begin{lemma}[\bf Indistinguishability]
\label{lem:construction-indistinguishability}
For every randomized first-order algorithm $\mathrm{Alg}$
with query budget $T<n/2$,
\begin{equation}
    \mathbb E_{A\sim\mathcal D_{+1}}
        [\rho_A^{\mathrm{Alg}}]
    =
    \mathbb E_{A\sim\mathcal D_{-1}}
        [\rho_A^{\mathrm{Alg}}].
    \label{eq:output-indistinguishability}
\end{equation}
\end{lemma}

\paragraph{Why the target mixtures remain separated.}
Although low query budget algorithm cannot detect the determinant sign,
the norm of a target sample carries information about it.
Comparing the expectations of the radial test function $x\mapsto e^{-\lambda\|x\|^2}$ under the two
single-block mixtures gives $\mathrm{TV}(P_{+1},P_{-1})\ge\frac12e^{-5n/\sqrt{\kappa}}$.
A sample from the full target consists of $k$ blocks that are independent given the common sign, and
tensorization of the Hellinger affinity amplifies the single-block separation to the bound below.
Adding blocks does not help the algorithm, since a single query already acts on all blocks and
Lemma~\ref{lem:construction-indistinguishability} holds for every $k$.
We summarize this property in
Lemma~\ref{lem:construction-separation},
whose proof is given in Section~\ref{sec:mixture-separation}.

\begin{lemma}[\bf Separation]
\label{lem:construction-separation}
The averaged targets satisfy
\begin{equation}
    \operatorname{TV}(\bar\pi_{+1},\bar\pi_{-1})
    \ge
    1-\exp\left\{
        -\frac{k}{8}
        \exp\left(-\frac{10n}{\sqrt{\kappa}}\right)
    \right\}.
    \label{eq:averaged-target-separation}
\end{equation}
\end{lemma}

\paragraph{Why the construction satisfies the regularity assumptions.}
The symmetric part $(O+O^\top)/2$ has spectrum in $[-1,1]$.
The affine transformation in Eq.~\eqref{eq:single-block-precision}
therefore places the spectrum of $A_O$ in $[1,\kappa]$.
Forming a block-diagonal matrix and adding an identity block
preserves these bounds. 
This gives the following lemma with proof in Section~\ref{sec:regularity}.

\begin{lemma}[\bf Regularity]
\label{lem:construction-admissibility}
For every $\boldsymbol O\in O(n)^k$,
$A_{\boldsymbol O}\in\mathcal A_{d,\kappa}$.
The potential $f_{A_{\boldsymbol O}}$ has unique minimizer $0$
and minimum value $0$.
\end{lemma}

\subsection{Proof of Theorem~\ref{thm:gaussian-lower-bound}}
\label{subsec:main-proof}

Let $\mathrm{Alg}$ have fixed query budget $T$ and satisfy
$\operatorname{TV}(\rho_A^{\mathrm{Alg}},\pi_A)\le\varepsilon$
for every $A\in\mathcal A_{d,\kappa}$.

\paragraph{Zero-query cases.}
A zero-query algorithm has the same output distribution on
$A_0=I_d$ and $A_1=\operatorname{diag}(\kappa,1,\ldots,1)$.
However, evaluating the event $\{|x_1|\le1/2\}$ and using
$\kappa\ge4$ gives
\[
    \operatorname{TV}(\pi_{A_0},\pi_{A_1})
    \ge
    2\int_{1/2}^{1}\frac{e^{-t^2/2}}{\sqrt{2\pi}}\,dt
    \ge \frac{e^{-1/2}}{\sqrt{2\pi}}
    > \frac18
    \ge 2\varepsilon,
\]
contradicting the triangle inequality.
Thus $T\ge1$.
If $2T+1>d$, integrality gives $T\ge d/2$, which proves
the claimed bound.
Henceforth assume $2T+1\le d$.

\paragraph{Comparing the two priors.}
Choose $n:=2T+1$ and $k:=\lfloor d/n\rfloor$.
Then $3\le n\le d$ and $2T<n$, so
Lemma~\ref{lem:construction-indistinguishability}
gives a common averaged output distribution $\bar\rho$
under the two priors.
By Lemma~\ref{lem:construction-admissibility} and convexity
of total variation under mixing,
\[
    \operatorname{TV}(\bar\rho,\bar\pi_b)
    \le
    \mathbb E_{A\sim\mathcal D_b}
    \left[
        \operatorname{TV}(\rho_A^{\mathrm{Alg}},\pi_A)
    \right]
    \le\varepsilon,
    \qquad b\in\{-1,+1\}.
\]
The triangle inequality and
Lemma~\ref{lem:construction-separation} therefore imply
\[
    1-\exp\left\{
        -\frac{k}{8}e^{-10n/\sqrt{\kappa}}
    \right\}
    \le
    \operatorname{TV}(\bar\pi_{+1},\bar\pi_{-1})
    \le2\varepsilon.
\]
Rearranging yields
\begin{equation}
    k e^{-10n/\sqrt{\kappa}}
    \le -8\log(1-2\varepsilon)
    \le32\varepsilon,
    \label{eq:block-error-constraint}
\end{equation}
where the last inequality uses
$-\log(1-2\varepsilon)
=\int_0^{2\varepsilon}(1-t)^{-1}\,dt
\le4\varepsilon$.

\paragraph{Bounding the query budget.}
Since $k\ge d/(2n)$ and $n=2T+1\le3T$,
Eq.~\eqref{eq:block-error-constraint} gives
\begin{equation}
    \frac{d}{64\varepsilon}
    \le n e^{10n/\sqrt{\kappa}}
    \le3T e^{30T/\sqrt{\kappa}}.
    \label{eq:exponential-budget-constraint}
\end{equation}
Using $\sinh t\ge t$ for $t\ge0$, we obtain
\[
    e^{60T/\sqrt{\kappa}}-1
    =
    2e^{30T/\sqrt{\kappa}}
    \sinh\left(\frac{30T}{\sqrt{\kappa}}\right)
    \ge
    \frac{60T}{\sqrt{\kappa}}e^{30T/\sqrt{\kappa}}
    \ge
    \frac{5d}{16\sqrt{\kappa}\varepsilon}.
\]
Taking logarithms and using concavity of the logarithm,
\[
    T
    \ge
    \frac{\sqrt{\kappa}}{60}
    \log\left(1+\frac{5d}{16\sqrt{\kappa}\varepsilon}\right)
    \ge
    \frac{\sqrt{\kappa}}{192}
    \log\left(1+\frac{d}{\sqrt{\kappa}\varepsilon}\right).
\]
The second inequality follows by writing
$1+5d/(16\sqrt{\kappa}\varepsilon)
=(11/16)+(5/16)(1+d/(\sqrt{\kappa}\varepsilon))$.

Combining this with the alternative $T\ge d/2$ yields
\[
    T\ge
    \frac1{256}
    \min\left\{
        d,\,
        \sqrt{\kappa}
        \log\left(1+\frac{d}{\sqrt{\kappa}\varepsilon}\right)
    \right\}.
\]
Since $\mathrm{Alg}$ was arbitrary, this proves the bound
for $Q_G(d,\kappa,\varepsilon)$.


\section{Lower Bound via Hidden Eigenbasis}
\label{sec:5}

The preprint~\citep{OAI:Subpolynomial-query-complexity-for-well-conditioned-log-concave-sampling-September-26-2026} proved an $\Omega(\log d)$ query lower bound for centered Gaussian
sampling with $I_d\preceq A\preceq2I_d$ at total variation accuracy $\varepsilon=\frac{1}{10}$.
%
%
We subsequently examine the limitations of this analysis for establishing accuracy-dependent lower bounds in Section~\ref{subsec:openai-limitations}.

\begin{theorem}[\bf Extension to general $\kappa$]
\label{thm:openai-general-kappa}
For every fixed $\delta\in(0,1)$, there exists $d_0(\delta)$ such
that, for every integer $d\ge d_0(\delta)$,
$2\le\kappa\le d^{1-\delta}$, and $0<\varepsilon\le1/10$,
\[
Q_G(d,\kappa,\varepsilon)
\ge \frac{\delta}{256}\sqrt{\kappa}\log d.
\]
\end{theorem}
Compared with Theorem~\ref{thm:gaussian-lower-bound}, this bound also covers $2\le\kappa<4$ and
$\frac{1}{16}<\varepsilon\le\frac{1}{10}$, and it holds even when the spectrum of $A$ is known to
the algorithm. On the other hand, it does not grow as $\varepsilon\to0$ and it requires $\kappa\le
d^{1-\delta}$.

\subsection{The Hard Family}
\label{subsec:openai-hard-family}
We use a construction similar to that in~\cite[Section~9.2]{OAI:Subpolynomial-query-complexity-for-well-conditioned-log-concave-sampling-September-26-2026},
with a cosine eigenvalue grid (Chebyshev nodes) and a random orthogonal rotation, while extending the spectral
interval from $[1,2]$ to $[1,\kappa]$.
For $d\ge2$ and $\kappa\ge2$, define
\begin{equation}
\label{eq:openai-hard-spectrum}
\theta_i=\frac{(i-\tfrac12)\pi}{d},
\qquad
\lambda_i=\frac{\kappa+1}{2}
          +\frac{\kappa-1}{2}\cos\theta_i,
\qquad
\Lambda=\operatorname{diag}(\lambda_1,\ldots,\lambda_d).
\end{equation}
Draw $O\sim\operatorname{Haar}(\mathrm{O}(d))$, independently
of the algorithm's private randomness, and set $A=O\Lambda O^\top$.
The associated potential and target are
$f_A(x)=x^\top Ax/2$ and $\pi_A=N(0,A^{-1})$.

The spectrum is deterministic and may be supplied to the algorithm.
In particular, the Gaussian normalizing constant is known.
The unknown component is the random eigenbasis represented by $O$.

\begin{remark}[\bf Comparison with $\kappa=2$ case]
At $\kappa=2$, the spectrum in
Eq.~\eqref{eq:openai-hard-spectrum} recovers the original construction in \cite{OAI:Subpolynomial-query-complexity-for-well-conditioned-log-concave-sampling-September-26-2026}.
More generally, if $A_2$ denotes the original precision matrix
with the same rotation $O$, then
$A=(\kappa-1)A_2-(\kappa-2)I_d$.
Thus one matrix--vector reply for either matrix determines
the corresponding reply for the other.
The proof compares $\sum_i\cos(n\theta_i)\langle x,Oe_i\rangle^2$ with $\|x\|^2$, where $n=2q+1$.
For the output of a $q$-query algorithm, the ratio is at most $d^{-\Omega(1)}$ in absolute value
(Lemma~\ref{lem:openai-cosine}), whereas under the target the ratio of the means has absolute value
$e^{-\Theta(n/\sqrt{\kappa})}$ (Lemma~\ref{lem:openai-target-bias}).
Hence the test works for $n$ up to order $\sqrt{\kappa}\log d$, rather than $\log d$ as at
$\kappa=2$, which is the source of the factor $\sqrt{\kappa}$ in
Theorem~\ref{thm:openai-general-kappa}.
%
%
\end{remark}

\subsection{Properties of the Hard Family}
\label{subsec:openai-properties}
This subsection states the four lemmas used in the proof of Theorem~\ref{thm:openai-general-kappa}.
Lemma~\ref{lem:openai-rotation} places the output of any $q$-query algorithm in a block Krylov space generated by independent
Gaussian vectors,
Lemma~\ref{lem:openai-cosine} bounds a quadratic form with cosine weights
uniformly over this space,
Lemma~\ref{lem:openai-target-bias} evaluates the same quadratic form under the target, and
Lemma~\ref{lem:openai-admissibility} verifies that the hard family is admissible.
All proofs are deferred to Appendix~\ref{app:c}.

\paragraph{Krylov representation of the output.}
We first control the information revealed by adaptive queries via a standard Krylov argument.
The deferred-rotation argument of~\cite[Lemma~9.2]{OAI:Subpolynomial-query-complexity-for-well-conditioned-log-concave-sampling-September-26-2026} uses the rotational symmetry
of the undisclosed eigenbasis to represent the algorithm's
output in a block Krylov space generated by
independent Gaussian vectors.

For vectors $g_1,\ldots,g_{q+1}\in\mathbb R^d$ and $\Lambda=\operatorname{diag}(\lambda_1,\ldots,\lambda_d)$, define
\[
\mathcal K_q(g_1,\ldots,g_{q+1})
=
\operatorname{span}\{
\Lambda^j g_a:
1\le a\le q+1,\ 0\le j\le q
\}.
\]

\begin{lemma}[\bf Krylov representation]
\label{lem:openai-rotation}
Fix an integer $q\ge0$ with $2q+1<d$ and let $X$ be the output. There exist independent vectors
$g_1,\ldots,g_{q+1}\sim N(0,I_d)$ such that
\[
O^\top X\in\mathcal K_q(g_1,\ldots,g_{q+1})
\qquad\text{almost surely}.
\]
\end{lemma}

The additional starting vector accounts for the 
output, which may lie outside the span of the query points
and oracle replies.
Equivalently, the lemma gives a representation
$O^\top X=\sum_{a=1}^{q+1}p_a(\Lambda)g_a$,
where each $p_a$ has degree at most $q$.
The polynomial coefficients may depend on the entire
extended experiment.
We therefore control the output through an estimate that
holds simultaneously for every vector in $\mathcal K_q$.
The proof of Lemma~\ref{lem:openai-rotation} is deferred to
Appendix~\ref{app:openai-properties}.

\paragraph{A uniform bound on the Krylov space.}
We next control a high-frequency cosine statistic uniformly over
the space in Lemma~\ref{lem:openai-rotation}.
For a query budget $q$, set
\[
n=2q+1,
\qquad
D_n=\operatorname{diag}
  \bigl(\cos(n\theta_1),\ldots,\cos(n\theta_d)\bigr).
\]
The choice of $n$ places this statistic above all frequencies
arising from products of two degree-$q$ cosine polynomials.

\begin{lemma}[\bf Uniform bound on the Krylov space]
\label{lem:openai-cosine}
There is an absolute constant $C>0$ such that the following holds.
Let $q\ge0$ satisfy $2q+1<d$, let
$g_1,\ldots,g_{q+1}$ be independent $N(0,I_d)$ vectors, and let
$\mathcal K_q=\mathcal K_q(g_1,\ldots,g_{q+1})$.
Define
\[
\tau_{d,q}
=
C\left[
(q+1)\sqrt{\frac{\log d}{d}}
+
\frac{(q+1)(q+1+\log d)\log d}{d}
\right].
\]
If $\tau_{d,q}\le1/2$, then, with probability at least $0.99$,
\[
|u^\top D_nu|
\le 2\tau_{d,q}\|u\|^2
\qquad
\text{for every }u\in\mathcal K_q.
\]
\end{lemma}

The proof is deferred to
Appendix~\ref{app:openai-cosine}.
Since $\lambda_i$ is a nonconstant affine function of
$\cos\theta_i$, the space $\mathcal K_q$ is equivalently spanned
by the vectors
$\bigl((g_a)_i\cos(j\theta_i)\bigr)_{i=1}^d$,
with $a\le q+1$ and $j\le q$.
Consequently, the estimate is independent of $\kappa$.

On the event in Lemma~\ref{lem:openai-cosine},
the output representation in Lemma~\ref{lem:openai-rotation}
therefore gives
\[
\bigl|X^\top OD_nO^\top X\bigr|
\le 2\tau_{d,q}\|X\|^2.
\]

\paragraph{Target moments.}
We now evaluate the same cosine statistic under the target Gaussian.

\begin{lemma}[\bf Moments under the target]
\label{lem:openai-target-bias}
Let $\kappa\ge2$, $U\sim N(0,\Lambda^{-1})$, and
$\varrho_\kappa=\tfrac{\sqrt\kappa-1}{\sqrt\kappa+1}$.
For an integer $1\le n<d$, put
$D_n=\operatorname{diag}(\cos(n\theta_i))$,
$M=\mathbb E\|U\|^2$, and
$\mu_n=\mathbb E[U^\top D_nU]$.
Then
\begin{align*}
M
&=\frac{d}{\sqrt\kappa}
  \frac{1-\varrho_\kappa^{2d}}
       {1+\varrho_\kappa^{2d}},\quad\mbox{and }
(-1)^n\mu_n
=\frac{d}{\sqrt\kappa}
  \frac{\varrho_\kappa^n-\varrho_\kappa^{2d-n}}
       {1+\varrho_\kappa^{2d}}.
\end{align*}
Moreover, both $\operatorname{Var}(\|U\|^2)$ and
$\operatorname{Var}(U^\top D_nU)$ are at most
$2M\le2d/\sqrt\kappa$.

If $d-n\ge\sqrt\kappa$, then
\[
(1-e^{-4})e^{-3n/\sqrt\kappa}
\le\frac{|\mu_n|}{M}
\le e^{-2n/\sqrt\kappa}.
\]
\end{lemma}

The proof is deferred to
Appendix~\ref{app:openai-target-bias}.
For an exact sample $Y\mid O\sim\pi_A$, the vector $U=O^\top Y$ has law $N(0,\Lambda^{-1})$, with
$\|Y\|^2=\|U\|^2$ and $Y^\top OD_nO^\top Y=U^\top D_nU$.
The lemma therefore gives the means and variances of $\|Y\|^2$ and $Y^\top OD_nO^\top Y$ under the target.

\paragraph{Admissibility and oracle equivalence.}
We finally verify that the hard family belongs to the prescribed
class. 
The proof is deferred to
Appendix~\ref{app:openai-admissibility}.

\begin{lemma}[\bf Admissibility]
\label{lem:openai-admissibility}
For every $O\in\mathrm O(d)$, the matrix $A=O\Lambda O^\top$
satisfies $I_d\preceq A\preceq\kappa I_d$.
\end{lemma}
\subsection{Proof of Theorem~\ref{thm:openai-general-kappa}}
\label{subsec:openai-main-proof}

We compare the probabilities of a single test event
under algorithmic outputs and target samples.

\begin{proof}[Proof of Theorem~\ref{thm:openai-general-kappa}]
Fix $\delta\in(0,1)$, $2\le\kappa\le d^{1-\delta}$,
and $0<\varepsilon\le1/10$.
Suppose that an algorithm uses at most
$q\le(\delta/256)\sqrt\kappa\log d$ queries and satisfies
$\operatorname{TV}(\rho_A^{\mathrm{Alg}},\pi_A)\le\varepsilon$
for every admissible $A$.
Lemma~\ref{lem:openai-admissibility} makes the hard family
admissible.

Set $n=2q+1$ and $\tau=\tau_{d,q}$.
Define statistic $S_A$ and the associated
test event $E_A$ as
\begin{align*}
S_A(x)
&:=(-1)^n x^\top
T_n\!\left(\frac{2A-(\kappa+1)I_d}{\kappa-1}\right)x
-2\tau\|x\|^2,\\
E_A&:=\{x:S_A(x)>0\},
\end{align*}
where $T_n$ is the degree-$n$ Chebyshev polynomial of the
first kind.

Geometrically, for $x\ne0$, the weights
$w_i(x)=|\langle x,Oe_i\rangle|^2/\|x\|^2$
describe the fractions of its energy $\|x\|^2$ carried by the
eigenvectors $Oe_i$ of $A$.
Using $T_n(\cos\theta_i)=\cos(n\theta_i)$, the event
$x\in E_A$ is equivalent to
$(-1)^n\sum_i w_i(x)\cos(n\theta_i)>2\tau$.
So $E_A$ is the event that the cosine coeﬀicient
$\sum_i w_i(x) \cos(n\theta_i)$ of these weights at frequency $n$
exceeds $2\tau$ in absolute value and has the sign that it
has under the target.

\paragraph{Algorithmic outputs rarely pass the test.}
The bound on $q$ gives
$q+1\le d^{(1-\delta)/2}\log d$, so $n=o(d)$.
Substituting this bound into the definition of $\tau$ yields
\[
\tau \lesssim
d^{-\delta/2}(\log d)^{3/2}
+d^{-\delta}(\log d)^3
\lesssim d^{-\delta/2}(\log d)^3
\longrightarrow0
\]
Thus $n<d$ and $\tau\le1/2$ for sufficiently large $d$
depending only on $\delta$, and the Lemma~\ref{lem:openai-rotation} and Lemma~\ref{lem:openai-cosine} apply.
Let $X$ be the algorithmic output on the random input
$A=O\Lambda O^\top$.
Lemma~\ref{lem:openai-rotation} places $O^\top X$ in
$\mathcal K_q$ on an extension preserving the
joint law of $(O,X)$.
Lemma~\ref{lem:openai-cosine} then gives
\[|X^\top OD_nO^\top X|\le2\tau\|X\|^2\]
with probability at least $0.99$.
On this event, $S_A(X)\le0$.
Consequently,
\[
\mathbb E_A[\rho_A^{\mathrm{Alg}}(E_A)]\le0.01.
\]

\paragraph{Target samples typically pass the test.}
Let $M=\operatorname{tr}(\Lambda^{-1})$ denote the target
mean squared norm, and let
$\mu_n=\operatorname{tr}(D_n\Lambda^{-1})$ denote the target
mean of the cosine statistic.
Write $\beta=|\mu_n|/M$ for the normalized target bias.

Since $n=o(d)$ and
$\sqrt\kappa\le d^{(1-\delta)/2}=o(d)$,
we have $d-n\ge\sqrt\kappa$ for sufficiently large $d$.
Lemma~\ref{lem:openai-target-bias} and the the bound on $q$ therefore give
\[
\beta
\ge(1-e^{-4})e^{-3n/\sqrt\kappa}
\ge b\,d^{-\delta/4},
\qquad b=(1-e^{-4})e^{-3},
\]
where we used
$3n/\sqrt\kappa\le3+(3\delta/128)\log d$.
Moreover,
$\varrho_\kappa\le e^{-1/\sqrt\kappa}$ and
$d\ge\sqrt\kappa$ imply
$\varrho_\kappa^{2d}\le e^{-2}<1/3$.
The formula for $M$ in the same lemma yields
$M\ge d/(2\sqrt\kappa)\ge d^{(1+\delta)/2}/2$.
Together with the bound on $\tau$ above, these estimates imply
\[
\frac{\tau}{\beta}
\le\frac{C'}{b}d^{-\delta/4}(\log d)^3\longrightarrow0,
\qquad
M\beta^2\ge\frac{b^2}{2}\sqrt d\longrightarrow\infty.
\]
All estimates are uniform over the stated parameter range.
We may therefore choose $d_0(\delta)$ so that the conditions
for the algorithmic bound above hold and, additionally,
$\tau\le\beta/4$ and $M\beta^2\ge320$.

Now fix $A$, draw $Y\sim\pi_A$, and set $Z=S_A(Y)$.
Lemma~\ref{lem:openai-target-bias} gives mean
$(-1)^n\mu_n=\beta M$ for the signed cosine quadratic form
and mean $M$ for $\|Y\|^2$.
Both variances are at most $2M$.
Hence, using the $L^2$ triangle inequality for centered variables,
\[
\mathbb E Z=(\beta-2\tau)M\ge\frac{\beta M}{2},
\qquad
\operatorname{Var}(Z)
\le\bigl(\sqrt{2M}+2\tau\sqrt{2M}\bigr)^2
\le8M.
\]
Chebyshev's inequality gives, for every $A$ in the hard family,
\[
\pi_A(E_A)
=\mathbb P(Z>0)
\ge1-\frac{\operatorname{Var}(Z)}{(\mathbb E Z)^2}
\ge1-\frac{32}{M\beta^2}
\ge0.9.
\]

\paragraph{Contradiction.} Finally, the per-input TV guarantee implies
$\pi_A(E_A)-\rho_A^{\mathrm{Alg}}(E_A)\le\varepsilon$.
Averaging over the random input and using the two probability
bounds gives
\[
0.89
\le\mathbb E_A\!\left[
\pi_A(E_A)-\rho_A^{\mathrm{Alg}}(E_A)
\right]
\le\varepsilon
\le0.1,
\]
a contradiction.
\end{proof}
\subsection{Limitations of the Analysis}
\label{subsec:openai-limitations}

Theorem~\ref{thm:openai-general-kappa} is restricted to
$2\le\kappa\le d^{1-\delta}$ for fixed $\delta\in(0,1)$.
Although it applies to every $\varepsilon\le1/10$,
its lower bound does not grow as $\varepsilon$ decreases.
We discuss two limitations of the analysis that affect
extensions to larger condition numbers and higher accuracy.

\paragraph{Limitation on $\kappa$.}
Our approach does not extend the $\Omega(\sqrt{\kappa}\log d)$ lower bound to condition numbers
$\kappa\gtrsim d/\log^2d$.
Our proof controls the output of the algorithm by first applying
Lemma~\ref{lem:openai-rotation} to place $O^\top X$ in
$\mathcal K_q$, and then bounding the cosine statistic
uniformly over this space.
However, the enlarged space  $\mathcal{K}_q$ has $(q+1)^2$ generating vectors.
Since the eigenvalues of $\Lambda$ are distinct,
these vectors span a space of the following dimension.

\begin{lemma}[\bf Dimension of the Krylov space]
\label{lem:openai-krylov-rank}
For every integer $q\ge0$ and independent
$g_1,\ldots,g_{q+1}\allowbreak\sim N(0,I_d)$,
\[
\dim\mathcal K_q
=\min\{d,(q+1)^2\}
\qquad\text{almost surely}.
\]
\end{lemma}

Once $(q+1)^2\ge d$ with $q =\Omega(\sqrt{\kappa}\log d)$,  we have
$\mathcal K_q=\mathbb R^d$ almost surely.
For $n=2q+1<d$, the definition of $\theta_i$ gives
$\sum_i^d\cos^2(n\theta_i)=d/2$, so at least one coordinate
satisfies $|\cos(n\theta_i)|\ge1/\sqrt2$.
Taking its coordinate unit vector yields
\[
\sup_{u\in\mathcal K_q\setminus\{0\}}
\frac{|u^\top D_nu|}{\|u\|^2}
\ge\frac1{\sqrt2}.
\]
Our separating test in Lemma~\ref{lem:openai-cosine} requires this supremum to be at most
$1/2$, which is incompatible with the displayed bound.
Thus an argument based on a uniform bound over $\mathcal K_q$
cannot apply  once
$(q+1)^2\ge d$.
%

\paragraph{Limitation on accuracy.}
Our test compares the ratio $(-1)^nx^\top OD_nO^\top x/\|x\|^2$ with a threshold $t$.
Lemma~\ref{lem:openai-cosine} puts every output of the algorithm below $t$ only if $t\ge a$, where $a$ is the
uniform bound on $\mathcal{K}_q$ that the lemma provides. Chebyshev's inequality with
Lemma~\ref{lem:openai-target-bias} puts samples from the target above $t$ only if $t<|\mu_n|/M$. Thus
$|\mu_n|/M>a$, whatever the threshold.
At the same time, the Bernstein calculation in the proof of
Fact~\ref{fact:openai-gram-concentration}
requires $a\ge c/\sqrt d$ even to certify a failure
probability below one.
Finally, Lemma~\ref{lem:openai-target-bias} gives
$|\mu_n|/M\le e^{-2n/\sqrt\kappa}$.
Since $n=2q+1$, we must have
\[
\frac{c}{\sqrt d}
\le a
<\frac{|\mu_n|}{M}
\le e^{-2(2q+1)/\sqrt\kappa}.
\]
Taking logarithms yields $q=O(\sqrt\kappa\log d)$,
independently of $\varepsilon$.
A different choice of the threshold or of the constants therefore cannot give a bound that grows as
$\varepsilon\to0$.

\begin{remark}[\bf Scope of the limitations]
These limitations concern our proof. We do not know whether the same fixed-spectrum family supports
stronger lower bounds.
Our analysis therefore leaves open whether the same
fixed-spectrum family supports stronger lower bounds.
\end{remark}
\section{Conclusion}

We proved an $\Omega\bigl(\min\{d,\sqrt{\kappa}\log(1+\frac{d}{\sqrt{\kappa}\varepsilon})\}\bigr)$ query
lower bound for Gaussian sampling.
The bound comes from a construction with a hidden determinant sign and matches the known upper bound
whenever $\kappa\le d^{2-\delta}$ or $\varepsilon\le d^{-c}$, for fixed $\delta\in(0,2)$ and $c>0$.
We also extended the hidden-eigenbasis construction to growing
condition numbers.

Closing the remaining gap for all condition numbers $\kappa\le d^{2-\delta}$
at constant accuracy remains open.
Our lower bounds also apply to strongly log-concave sampling, where the known upper bounds are much
larger. It is open whether non-quadratic potentials are harder than Gaussians.


\section*{AI Disclosure}
AI tools were used to assist with literature searches, the development and verification of proof
arguments, the presentation of proofs, the polishing of the exposition, and typesetting. The authors
independently verified all mathematical arguments and take full responsibility for the content of the
paper.

\section*{Acknowledgment}
We thank Anming Gu and Taoli Zheng for helpful discussions. 
TS was partially supported by
JST CREST (JPMJCR2115) and JSPS KAKENHI (24K02905). 
This work was partially supported by JST ERATO (JPMJER2601). 
This research is supported by the
National Research Foundation, Singapore and the Ministry of Digital Development and Information
under the AI Visiting Professorship Programme (award number AIVP-2024-004). Any opinions,
findings and conclusions or recommendations expressed in this material are those of the author(s)
and do not reflect the views of National Research Foundation, Singapore and the Ministry of Digital
Development and Information.

\newpage
\bibliography{reference}
\bibliographystyle{alpha}

\newpage

\appendix

\section{Useful Facts}
\subsection{Two Spherical Poisson Identities}

Both identities below are instances of the Poisson integral formula for the
unit ball of $\mathbb R^n$, applied at the point $re_1$ to the harmonic
functions $1$ and $y\mapsto y_1$.
\begin{lemma}[\bf Spherical Poisson identities]
\label{app:spherical-poisson}
\label{lem:spherical-poisson}
Let $n\ge2$, and $\zeta=(\zeta_1,\dots,\zeta_n)$ be a uniform random point on
$\mathbb S^{n-1}$, and let $U=\zeta_1$ be its first coordinate., where $n\ge3$.
For every $0<r<1$,
\begin{equation}
    \mathbb E\left[(1+r^2-2rU)^{-n/2}\right]
    =\frac1{1-r^2},
    \qquad
    \mathbb E\left[U(1+r^2-2rU)^{-n/2}\right]
    =\frac r{1-r^2}.
    \label{eq:spherical-poisson}
\end{equation}
\end{lemma}

\begin{proof}
Let $\sigma$ be normalized surface measure on $\mathbb S^{n-1}$.
By the Poisson representation formula
\cite[Theorem~1.21]{axler2001harmonic},
every function $h$ harmonic on the unit ball and continuous
on its closure satisfies
\[
    h(a)
    =
    \int_{\mathbb S^{n-1}}
        h(u)\frac{1-\|a\|^2}{\|u-a\|^n}\,d\sigma(u),
    \qquad \|a\|<1.
\]
Apply this formula with $a=re_1$ and, respectively,
$h(x)=1$ and $h(x)=x_1$.
Since $\|u-re_1\|^2=1+r^2-2ru_1$, the two applications prove the lemma
after division by $1-r^2$.
\end{proof}

\section{Proofs for Section~\ref{sec:construction}}
\label{app:b}
\subsection{Indistinguishability: Proof of Lemma~\ref{lem:construction-indistinguishability} }
\label{sec:indistinguishability}

In this section, we prove Lemma~\ref{lem:construction-indistinguishability}.
Throughout this section, $2T<n$.

\begin{lemma}[\bf Indistinguishability (Lemma 5 restated)]
For every randomized first-order algorithm $\mathrm{Alg}$
with query budget $T<n/2$,
\begin{equation}
    \mathbb E_{A\sim\mathcal D_{+1}}
        [\rho_A^{\mathrm{Alg}}]
    =
    \mathbb E_{A\sim\mathcal D_{-1}}
        [\rho_A^{\mathrm{Alg}}].
\end{equation}
\end{lemma}
\paragraph{Oracle reduction.}
Consider the oracle
\[
\mathcal O_{\mathbf O}(j,v,\tau)
:=
\begin{cases}
O_jv, & \tau=\mathsf f,\\
O_j^\top v, & \tau=\mathsf i,
\end{cases}
\qquad
j\in\{1,\ldots,k\},\quad v\in\mathbb R^n.
\]
We call a query $(j, v, \tau)$ to
this oracle a primitive query; it is forward if $\tau = \mathsf f$,
and inverse if $\tau=\mathsf i$.
For a first-order query $x=(x_1,\ldots,x_k,z)$, we have
\begin{align*}
A_{O_j}x_j
&=
\frac{\kappa+1}{2}x_j
-\frac{\kappa-1}{4}
\bigl(
\mathcal O_{\mathbf O}(j,x_j,\mathsf f)
+
\mathcal O_{\mathbf O}(j,x_j,\mathsf i)
\bigr),\\
\nabla f_{A_{\mathbf O}}(x)
&=
(A_{O_1}x_1,\ldots,A_{O_k}x_k,z),\\
f_{A_{\mathbf O}}(x)
&=
\frac12\langle x,\nabla f_{A_{\mathbf O}}(x)\rangle.
\end{align*}
Thus $T$ first-order queries can be simulated using at most $2T$
primitive queries per block and it suffices to prove transcript
indistinguishability for adaptive primitive queries when $2T<n$.

\paragraph{Single-observation Haar laws.}
When at least two dimensions remain, a forward or inverse observation
has a sign-independent law. The following lemma also identifies the
conditional Haar law of the remaining map.
\begin{lemma}[\bf Conditional law of $W  $given $Wu$]
\label{lem:haar-observation}
Let $m\ge2$, $b\in\{-1,+1\}$, and $W\sim\mu_{m,b}$.
Write $\sigma_{m-1}$ for uniform probability measure on $S^{m-1}$.
Choose Borel maps $z\mapsto Q_z\in SO(m)$ satisfying $Q_ze_1=z$.
For every $u\in S^{m-1}$,
\[
Wu\sim\sigma_{m-1},
\qquad
\mathcal L(W\mid Wu=z)
=
\mathcal L\!\left(
Q_z
\begin{pmatrix}1&0\\0&V\end{pmatrix}
Q_u^\top
\right),
\quad V\sim\mu_{m-1,b},
\]
for $\sigma_{m-1}$-almost every $z$.
The same identities hold with $W$ replaced by $W^\top$.
\end{lemma}

\begin{proof}
The Borel completions $Q_z$ can be obtained by fixed-order
Gram--Schmidt, followed by a sign correction on the last column.
Let $Z\sim\sigma_{m-1}$ and $V\sim\mu_{m-1,b}$ be independent, and set
$\widetilde W:=Q_Z\operatorname{diag}(1,V)Q_u^\top$.
Then $\widetilde Wu=Z$ and $\det\widetilde W=b$.

For $R\in SO(m)$,
\[
Q_{Rz}^\top RQ_z
=
\begin{pmatrix}1&0\\0&S_R(z)\end{pmatrix},
\qquad S_R(z)\in SO(m-1).
\]
Rotational invariance of $Z$ and Haar invariance of $V$ give
$(RZ,S_R(Z)V)\stackrel{d}{=}(Z,V)$, hence
\[
R\widetilde W
=
Q_{RZ}\operatorname{diag}(1,S_R(Z)V)Q_u^\top
\stackrel{d}{=}\widetilde W.
\]
Thus $\widetilde W\sim\mu_{m,b}$ by uniqueness of invariant
probability measure on the determinant component.
Since $\widetilde Wu=Z$, conditioning on $Z=z$ yields both identities.
Finally, $W^\top\sim\mu_{m,b}$, so the same argument applies to
inverse observations.
\end{proof}
\paragraph{Adaptive posterior laws.}
We now apply Lemma~\ref{lem:haar-observation} to adaptive queries
across blocks, maintaining the conditional Haar law of each residual map.

\begin{lemma}[\bf Adaptive posterior factorization]
\label{lem:adaptive-posterior}
Fix $b\in\{-1,+1\}$ and let $O_1,\ldots,O_k$ be independent
with law $\mu_{n,b}$.
Consider a deterministic adaptive primitive-query algorithm making
at most $2T<n$ queries to each block.
Let $H_t$ be its complete history and $\ell_{j,t}$ the number of
queries to block $j$.
There exist history-measurable frames
$Q_{D,j,t},Q_{R,j,t}\in SO(n)$ and integers
$r_{j,t}\le\ell_{j,t}$ such that
\[
O_j
=
Q_{R,j,t}\operatorname{diag}(I_{r_{j,t}},W_{j,t})
Q_{D,j,t}^\top,
\qquad
\mathcal L_b\bigl((W_{j,t})_{j=1}^k\mid H_t=h\bigr)
=
\bigotimes_{j=1}^k\mu_{n-r_{j,t}(h),b}.
\]
The frames and exposed dimensions are chosen by rules independent of $b$.
\end{lemma}

\begin{proof}
At $t=0$, take
$Q_{D,j,0}=Q_{R,j,0}=I_n$, $r_{j,0}=\ell_{j,0}=0$,
and $W_{j,0}=O_j$.
The assertion follows from the independence of the blocks.

Assume the assertion holds at time $t$.
Fix a nonterminal history $h$ and its next query $(j,v,\tau)$.
For the selected block, suppress $j,t$ and write
$Q_D=[U_D,E]$, $Q_R=[U_R,F]$, and $m=n-r$.
The two query types give
\[
\begin{array}{lll}
\tau=\mathsf f:
&v=U_Da+Ec,
&Y=U_Ra+FWc,\\
\tau=\mathsf i:
&v=U_Ra+Fc,
&Y=U_Da+EW^\top c.
\end{array}
\]
Here $a\in\mathbb R^r$ and $c\in\mathbb R^m$ are the coordinates of $v$ in the corresponding frame: $a=U_D^\top v$, $c=E^\top v$ for a forward
query, and $a=U_R^\top v$, $c=F^\top v$ for an inverse query.
If $c=0$, the response $Y$ is determined by $h$, so the query reveals
nothing about $W$. In this case we leave the frames, the exposed
dimension, and the residual map unchanged.

Suppose $c\ne0$. The algorithm makes at most $2T$ queries to each block,
and the next query is addressed to block $j$, so $\ell_{j,t}\le2T-1$.
Together with $r\le\ell_{j,t}$ and $2T<n$, this gives
\[
m=n-r\ge n-2T+1\ge2.
\]
Write $c=\|c\|u$ with $u\in S^{m-1}$. For a forward query,
$Y=U_Ra+\|c\|\,F(Wu)$, so given $h$ the response $Y$ and the vector
$z:=Wu=F^\top Y/\|c\|$ determine each other. For an inverse query the
same holds with $z:=W^\top u=E^\top Y/\|c\|$.
By Lemma~\ref{lem:haar-observation}, conditionally on this response,
\[
\begin{array}{ll}
W=Q_z\operatorname{diag}(1,V)Q_u^\top,
&\tau=\mathsf f,\\
W^\top=Q_z\operatorname{diag}(1,V)Q_u^\top,
&\tau=\mathsf i,
\end{array}
\qquad V\sim\mu_{m-1,b}.
\]
Put $B_u=\operatorname{diag}(I_r,Q_u)$ and
$B_z=\operatorname{diag}(I_r,Q_z)$, and update
\[
(Q_D^+,Q_R^+,W^+)
=
\begin{cases}
(Q_DB_u,Q_RB_z,V), & \tau=\mathsf f,\\
(Q_DB_z,Q_RB_u,V^\top), & \tau=\mathsf i.
\end{cases}
\]
In both cases,
\[
O_j
=
Q_R^+\operatorname{diag}(I_{r+1},W^+)(Q_D^+)^\top,
\qquad
\mathcal L_b(W^+\mid H_t=h,Y)=\mu_{m-1,b}.
\]
All updated frames belong to $SO(n)$ and depend only on $(h,Y)$.

Leave the other blocks unchanged and set
$r_i^+=r_i+\mathbf1_{\{i=j\}}\mathbf1_{\{c\ne0\}}$.
These dimensions are determined by $h$.
Let $K_b(h,\cdot)$ be the conditional response law.
The induction hypothesis and the single-block update yield
\[
\mathcal L_b\bigl(Y,(W_i^+)_{i=1}^k\mid H_t=h\bigr)
=
K_b(h,\cdot)\otimes
\bigotimes_{i=1}^k\mu_{n-r_i^+,b}.
\]
Indeed, only the selected residual map enters the response,
and the other residual maps are conditionally independent of it.
Disintegrating with respect to $Y$, and using
$H_{t+1}=(H_t,(j,v,\tau),Y)$, proves the posterior identity
at time $t+1$.
Moreover, $r_i^+\le\ell_{i,t+1}$.

After stopping, pad the history with an absorbing symbol and
leave all states unchanged.
This completes the induction.
\end{proof}

\paragraph{Transcript indistinguishability.}
The residual Haar laws give a response kernel independent of the
determinant sign. Iterating this kernel yields identical transcript laws.

\begin{lemma}[\bf Transcript indistinguishability]
\label{lem:primitive-indistinguishability}
In the setting of Lemma~\ref{lem:adaptive-posterior},
the padded transcript $H_M$, with $M=2Tk$, satisfies
\[
\mathcal L_{+1}(H_M)=\mathcal L_{-1}(H_M).
\]
\end{lemma}

\begin{proof}
Fix a nonterminal history $h$ and its next query
$\alpha_t(h)=(j,v,\tau)$.
Use the history-measurable frames from
Lemma~\ref{lem:adaptive-posterior}, and write
$Q_D=[U_D,E]$, $Q_R=[U_R,F]$ for the selected block.
Its residual dimension satisfies
\[
m=n-r_{j,t}
\ge n-\ell_{j,t}
\ge n-2T+1
\ge2.
\]
By Lemmas~\ref{lem:adaptive-posterior}
and~\ref{lem:haar-observation}, the conditional response law
under either sign is the kernel
\[
K_t(h,B)
:=
\begin{cases}
\displaystyle
\int_{S^{m-1}}
\mathbf1_B\!\left(
U_RU_D^\top v+\|E^\top v\|Fz
\right)\,\sigma_{m-1}(dz),
& \tau=\mathsf f,\\[2mm]
\displaystyle
\int_{S^{m-1}}
\mathbf1_B\!\left(
U_DU_R^\top v+\|F^\top v\|Ez
\right)\,\sigma_{m-1}(dz),
& \tau=\mathsf i.
\end{cases}
\]
These formulas also cover zero residual query vectors.
At terminal histories, set
$\alpha_t(h)=\dagger$ and $K_t(h,\cdot)=\delta_\dagger$.
The resulting Borel kernel is independent of $b$ and
extend it arbitrarily outside legal histories.

Since $H_{t+1}=(H_t,\alpha_t(H_t),Y_{t+1})$,
for every measurable set $B$ of histories,
\[
\mathbb P_b(H_{t+1}\in B)
=
\int\!\int
\mathbf1_B(h,\alpha_t(h),y)\,
K_t(h,dy)\,\mathcal L_b(H_t)(dh).
\]
Thus $\mathcal L_{+1}(H_t)=\mathcal L_{-1}(H_t)$ implies
$\mathcal L_{+1}(H_{t+1})=\mathcal L_{-1}(H_{t+1})$.
Starting from the common empty history proves the assertion.
\end{proof}


\begin{proof}[Proof of Lemma~\ref{lem:construction-indistinguishability}]
Let $\xi$ be the algorithm's entire random seed.
For each fixed seed, the oracle reduction gives a deterministic
primitive-query algorithm with at most $2T<n$ queries per block.
Applying Lemma~\ref{lem:primitive-indistinguishability}
and integrating over the common seed law yields
\[
\mathcal L_{+1}(\xi,H_M)=\mathcal L_{-1}(\xi,H_M),
\]
where $M=2Tk$.
The output is a Borel function $\Psi(\xi,H_M)$,
so its law is identical under both signs. Hence
\[
\mathbb E_{A\sim D_{+1}}[\rho_A^{\mathrm{Alg}}]
=
\mathbb E_{A\sim D_{-1}}[\rho_A^{\mathrm{Alg}}].
\]
\end{proof}

\subsection{Separation: Proof of Lemma~\ref{lem:construction-separation}}
\label{sec:mixture-separation}

In this section, we prove Lemma~\ref{lem:construction-separation}.

\begin{lemma}[\bf Separation (Lemma 6 restated)]
The averaged targets satisfy
\begin{equation}
    \operatorname{TV}(\bar\pi_{+1},\bar\pi_{-1})
    \ge
    1-\exp\left\{
        -\frac{k}{8}
        \exp\left(-\frac{10n}{\sqrt{\kappa}}\right)
    \right\}.
\end{equation}
\end{lemma}

We first separate the single-block mixtures using a bounded radial
test, and then amplify this separation across independent blocks.
The single-block argument combines a cancellation identity for
Gaussian normalization factors with two spherical Poisson identities.
Throughout this section, $n\ge3$, $\kappa\ge4$, and $O$ is Haar
distributed on the full orthogonal group $O(n)$.
Write $U:=O_{11}$, let $p_O$ be the density of $N(0,A_O^{-1})$,
and let $p_b$ be the density of $P_b$, for $b\in\{-1,+1\}$.

\subsubsection{Determinant Cancellation}
\label{subsec:determinant-cancellation}

The difference between averages over the two determinant components
can be expressed as a Haar expectation weighted by $\det(O)$.
The next lemma evaluates this signed expectation when the integrand
contains the Gaussian normalization factor $\sqrt{\det A_O}$ and
otherwise depends only on $U$.

\begin{lemma}[\bf  Determinant cancellation]
\label{lem:determinant-cancellation}
For every Borel function $h$ with $\mathbb E|h(U)|<\infty$,
\begin{equation}
\begin{aligned}
    &\mathbb E\left[\det(O)\sqrt{\det A_O}\,h(U)\right]=
    \frac{(-1)^{n-1}(\sqrt{\kappa}-1)^{n-1}}{2^n}
    \mathbb E\left[
        \bigl((\sqrt{\kappa}+1)U-(\sqrt{\kappa}-1)\bigr)h(U)
    \right].
\end{aligned}
    \label{eq:determinant-cancellation}
\end{equation}
\end{lemma}

\begin{proof}
We first express the normalization factor as a determinant.
The definition of $A_O$ gives
\begin{equation}
    A_O
    =
    \frac14
    \bigl((\sqrt{\kappa}+1)I_n-(\sqrt{\kappa}-1)O\bigr)^\top
    \bigl((\sqrt{\kappa}+1)I_n-(\sqrt{\kappa}-1)O\bigr).
    \label{eq:precision-factorization}
\end{equation}
For $0\le t\le\sqrt{\kappa}-1$ and $v\in\mathbb R^n$,
\[
    \|((\sqrt{\kappa}+1)I_n-tO)v\|
    \ge(\sqrt{\kappa}+1-t)\|v\|
    \ge2\|v\|.
\]
Thus $(\sqrt{\kappa}+1)I_n-tO$ is invertible throughout this
interval. Its determinant is positive at $t=0$ and remains
positive by continuity. Hence Eq.~\eqref{eq:precision-factorization}
implies
\begin{equation}
    \sqrt{\det A_O}
    =
    \frac1{2^n}
    \det\bigl((\sqrt{\kappa}+1)I_n-(\sqrt{\kappa}-1)O\bigr).
    \label{eq:gaussian-normalization-factor}
\end{equation}
Multiplying the matrix inside the determinant by $O^\top$,
and using $\det(O^\top)=\det(O)$, yields
\begin{equation}
\begin{aligned}
    \det(O)\sqrt{\det A_O}
    &=
    \frac1{2^n}
    \det\bigl(
        (\sqrt{\kappa}+1)O^\top-(\sqrt{\kappa}-1)I_n
    \bigr)\\
    &=
    \frac1{2^n}
    \det\bigl(
        (\sqrt{\kappa}+1)O-(\sqrt{\kappa}-1)I_n
    \bigr),
\end{aligned}
    \label{eq:signed-normalization-factor}
\end{equation}
where the last equality follows by transposing.

We next identify the principal minors that survive averaging.
For $S\subseteq\{1,\ldots,n\}$, write $O_{S,S}$ for the
corresponding principal submatrix, with
$\det O_{\varnothing,\varnothing}=1$.
If $S$ contains an index $i\ge2$, the row reflection
$D_i:=I_n-2e_ie_i^\top$ satisfies
\[
    D_iO\overset{d}=O,\qquad
    (D_iO)_{11}=U,\qquad
    \det((D_iO)_{S,S})=-\det O_{S,S}.
\]
Consequently,
\[
\begin{aligned}
    \mathbb E[\det O_{S,S}h(U)]
    &=
    \mathbb E\left[
        \det((D_iO)_{S,S})h((D_iO)_{11})
    \right]\\
    &=
    -\mathbb E[\det O_{S,S}h(U)],
\end{aligned}
\]
so this expectation is zero.
All expectations are finite because $|\det O_{S,S}|\le1$.
Thus only $S=\varnothing$ and $S=\{1\}$ contribute when
we expand Eq.~\eqref{eq:signed-normalization-factor} in principal minors:
\[
\begin{aligned}
    &\mathbb E[\det(O)\sqrt{\det A_O}\,h(U)]\\
    &=
    \frac1{2^n}
    \sum_{S\subseteq\{1,\ldots,n\}}
    (\sqrt{\kappa}+1)^{|S|}
    \bigl(-(\sqrt{\kappa}-1)\bigr)^{n-|S|}
    \mathbb E[\det O_{S,S}h(U)]\\
    &=
    \frac1{2^n}\left[
        \bigl(-(\sqrt{\kappa}-1)\bigr)^n\mathbb E[h(U)]
        +(\sqrt{\kappa}+1)
        \bigl(-(\sqrt{\kappa}-1)\bigr)^{n-1}\mathbb E[Uh(U)]
    \right]\\
    &=
    \frac{(-1)^{n-1}(\sqrt{\kappa}-1)^{n-1}}{2^n}
    \mathbb E\left[
        \bigl((\sqrt{\kappa}+1)U-(\sqrt{\kappa}-1)\bigr)h(U)
    \right].
\end{aligned}
\]
\end{proof}

\subsubsection{Single-block Separation}
\label{subsec:radial-test}

We distinguish $P_{+1}$ and $P_{-1}$ using a test of the form
$x\mapsto e^{-\lambda\|x\|^2}$.
After Gaussian integration and determinant cancellation, its
expectation difference reduces to a spherical integral.
The following identities evaluate that integral.



Choosing $\lambda$ so that these identities apply gives the
following quantitative separation.

\begin{proposition}[\bf Single-block separation]
\label{prop:single-block-separation}
For $n\ge3$ and $\kappa\ge4$,
\begin{equation}
    \operatorname{TV}(P_{+1},P_{-1})
    \ge
    \frac{
        \left(\frac{\sqrt{\kappa}-1}{\sqrt{\kappa}+1}\right)^n
    }{
        2+\frac{4n\sqrt{\kappa}}{(\sqrt{\kappa}+1)^2}
    }
    \ge
    \frac12e^{-5n/\sqrt{\kappa}}.
    \label{eq:single-block-separation}
\end{equation}
\end{proposition}

\begin{proof}
For $\lambda\ge0$, let $h_\lambda(x):=e^{-\lambda\|x\|^2}$
and define
\[
    H(\lambda)
    :=
    \mathbb E\left[
        \det(O)\int_{\mathbb R^n}h_\lambda(x)p_O(x)\,dx
    \right].
\]
Since the two determinant components have Haar probability $1/2$,
\[
    H(\lambda)
    =
    \frac12\left(
        \mathbb E_{P_{+1}}h_\lambda
        -\mathbb E_{P_{-1}}h_\lambda
    \right).
\]
As $0<h_\lambda\le1$, it follows that
\begin{equation}
    \operatorname{TV}(P_{+1},P_{-1})\ge2|H(\lambda)|.
    \label{eq:single-block-test-gap}
\end{equation}
We compute $H(\lambda)$ and then choose $\lambda>0$.

For every $Q\in O(n)$, conjugation preserves Haar measure and
the determinant sign, while $A_{QOQ^\top}=QA_OQ^\top$ and
$p_{QOQ^\top}(Qx)=p_O(x)$.
Consequently, for every $x\in\mathbb R^n$,
\[
    \mathbb E[\det(O)p_O(x)]
    =
    \mathbb E[\det(O)p_O(\|x\|e_1)].
\]
At $\|x\|e_1$, the Gaussian exponent depends on $O$ only
through $U$, since
$e_1^\top A_Oe_1=(\kappa+1-(\kappa-1)U)/2$.
Moreover, $I_n\preceq A_O\preceq\kappa I_n$ gives the
integrable bound
\[
    h_\lambda(x)p_O(\|x\|e_1)
    \le
    \left(\frac{\kappa}{2\pi}\right)^{n/2}e^{-\|x\|^2/2}.
\]
Using $|\det(O)|=1$, Fubini's theorem, and Gaussian integration,
we obtain
\[
\begin{aligned}
    H(\lambda)
    &=
    \frac1{(2\pi)^{n/2}}
    \mathbb E\left[
        \det(O)\sqrt{\det A_O}
        \int_{\mathbb R^n}
        e^{-[\kappa+1+4\lambda-(\kappa-1)U]\|x\|^2/4}\,dx
    \right]\\
    &=
    2^{n/2}
    \mathbb E\left[
        \frac{\det(O)\sqrt{\det A_O}}
        {(\kappa+1+4\lambda-(\kappa-1)U)^{n/2}}
    \right],
\end{aligned}
\]
where $\int_{\mathbb R^n}e^{-a\|x\|^2/4}\,dx=(4\pi/a)^{n/2}$
for $a>0$.
The denominator is at least $2+4\lambda$, so
Lemma~\ref{lem:determinant-cancellation} applies with
$h(u)=(\kappa+1+4\lambda-(\kappa-1)u)^{-n/2}$ and gives
\begin{equation}
    H(\lambda)
    =
    \frac{(-1)^{n-1}(\sqrt{\kappa}-1)^{n-1}}{2^{n/2}}
    \mathbb E\left[
        \frac{(\sqrt{\kappa}+1)U-(\sqrt{\kappa}-1)}
        {(\kappa+1+4\lambda-(\kappa-1)U)^{n/2}}
    \right].
    \label{eq:radial-laplace-transform}
\end{equation}
Since $Oe_1$ is uniform on $\mathbb S^{n-1}$,
Lemma~\ref{lem:spherical-poisson} applies to $U$.

We now match the denominator with a positive multiple of
$1+r^2-2rU$.
At $\lambda=0$, this matching gives
$r=(\sqrt{\kappa}-1)/(\sqrt{\kappa}+1)$, but $H(0)=0$
because both target distributions have total mass one.
We therefore decrease $r$ by a factor $1-1/n$ and set
\begin{equation}
\begin{aligned}
    r
    &:=
    \left(1-\frac1n\right)
    \frac{\sqrt{\kappa}-1}{\sqrt{\kappa}+1},\\
    \lambda
    &:=
    \frac{\kappa-1}{8r}(1+r^2)-\frac{\kappa+1}{4}
    =
    \frac{4n\sqrt{\kappa}+(\sqrt{\kappa}-1)^2}{8n(n-1)}.
\end{aligned}
    \label{eq:radial-test-parameters}
\end{equation}
Then $0<r<1$, $\lambda>0$, and
\begin{equation}
\begin{aligned}
    \kappa+1+4\lambda-(\kappa-1)U
    &=
    \frac{\kappa-1}{2r}(1+r^2-2rU)\\
    &=
    \frac{n(\sqrt{\kappa}+1)^2}{2(n-1)}(1+r^2-2rU).
\end{aligned}
    \label{eq:poisson-denominator}
\end{equation}
The two Poisson identities yield
\[
\begin{aligned}
    \mathbb E\left[
        \frac{(\sqrt{\kappa}+1)U-(\sqrt{\kappa}-1)}
             {(1+r^2-2rU)^{n/2}}
    \right]
    &=
    \frac{(\sqrt{\kappa}+1)r-(\sqrt{\kappa}-1)}{1-r^2}\\
    &=
    -\frac{\sqrt{\kappa}-1}{n(1-r^2)}.
\end{aligned}
\]
Substitution into Eq.~\eqref{eq:radial-laplace-transform} gives
\begin{equation}
\begin{aligned}
    H(\lambda)
    &=
    \frac{(-1)^{n-1}(\sqrt{\kappa}-1)^{n-1}}
         {(\sqrt{\kappa}+1)^n}
    \left(1-\frac1n\right)^{n/2}
    \left(-\frac{\sqrt{\kappa}-1}{n(1-r^2)}\right)\\
    &=
    \frac{(-1)^n}{n}
    \left(\frac{\sqrt{\kappa}-1}{\sqrt{\kappa}+1}\right)^n
    \frac{(1-1/n)^{n/2}}{1-r^2}.
\end{aligned}
    \label{eq:exact-radial-transform}
\end{equation}

It remains to lower-bound the test gap.
Since $n/2\ge1$, convexity of $t\mapsto t^{n/2}$ implies
\[
    \left(1-\frac1n\right)^{n/2}
    \ge1-\frac n2\cdot\frac1n=\frac12.
\]
Also,
\[
\begin{aligned}
    n(1-r^2)
    &=
    \frac{4n\sqrt{\kappa}}{(\sqrt{\kappa}+1)^2}
    +
    \left(\frac{\sqrt{\kappa}-1}{\sqrt{\kappa}+1}\right)^2
    \left(2-\frac1n\right)\\
    &\le
    \frac{4n\sqrt{\kappa}}{(\sqrt{\kappa}+1)^2}+2.
\end{aligned}
\]
Together with Eq.~\eqref{eq:single-block-test-gap} and
Eq.~\eqref{eq:exact-radial-transform}, these estimates give
\[
\begin{aligned}
    \operatorname{TV}(P_{+1},P_{-1})
    &\ge
    \frac2n
    \left(\frac{\sqrt{\kappa}-1}{\sqrt{\kappa}+1}\right)^n
    \frac{(1-1/n)^{n/2}}{1-r^2}\\
    &\ge
    \frac{
        \left(\frac{\sqrt{\kappa}-1}{\sqrt{\kappa}+1}\right)^n
    }{
        2+\frac{4n\sqrt{\kappa}}{(\sqrt{\kappa}+1)^2}
    }.
\end{aligned}
\]
This proves the first claimed bound.
For the second, $\kappa\ge4$ implies
\[
\begin{aligned}
    \log\left(\frac{\sqrt{\kappa}+1}{\sqrt{\kappa}-1}\right)
    &=
    2\int_0^{1/\sqrt{\kappa}}\frac{dt}{1-t^2}\\
    &\le
    \frac{2}{\sqrt{\kappa}(1-\kappa^{-1})}
    \le\frac{8}{3\sqrt{\kappa}}
    \le\frac3{\sqrt{\kappa}}.
\end{aligned}
\]
Hence
$((\sqrt{\kappa}-1)/(\sqrt{\kappa}+1))^n
\ge e^{-3n/\sqrt{\kappa}}$.
Using $1+t\le e^t$, we also have
\[
    2+\frac{4n\sqrt{\kappa}}{(\sqrt{\kappa}+1)^2}
    \le2+\frac{4n}{\sqrt{\kappa}}
    \le2e^{2n/\sqrt{\kappa}}.
\]
Therefore,
\[
    \operatorname{TV}(P_{+1},P_{-1})
    \ge
    \frac{e^{-3n/\sqrt{\kappa}}}{2e^{2n/\sqrt{\kappa}}}
    =
    \frac12e^{-5n/\sqrt{\kappa}}.
\]
\end{proof}

\subsubsection{Amplification Across Blocks}
\label{subsec:mixture-amplification}

Under each prior, the $k$ blocks of a target sample are independent
and have the same single-block mixture law.
We convert the preceding separation into an upper bound on
Hellinger affinity, which multiplies across independent blocks.

\begin{proof}[Proof of Lemma~\ref{lem:construction-separation}]
Let $a$ denote the single-block Hellinger affinity:
\[
    a:=\int_{\mathbb R^n}\sqrt{p_{+1}(x)p_{-1}(x)}\,dx.
\]
Cauchy--Schwarz and $\int p_{+1}=\int p_{-1}=1$ give $0\le a\le1$
and
\[
\begin{aligned}
    \operatorname{TV}(P_{+1},P_{-1})
    &=
    \frac12\int_{\mathbb R^n}
    |\sqrt{p_{+1}}-\sqrt{p_{-1}}|
    (\sqrt{p_{+1}}+\sqrt{p_{-1}})\\
    &\le
    \frac12
    \left(\int_{\mathbb R^n}
        (\sqrt{p_{+1}}-\sqrt{p_{-1}})^2\right)^{1/2}
    \left(\int_{\mathbb R^n}
        (\sqrt{p_{+1}}+\sqrt{p_{-1}})^2\right)^{1/2}\\
    &=
    \frac12\sqrt{(2-2a)(2+2a)}
    =
    \sqrt{1-a^2}.
\end{aligned}
\]
Proposition~\ref{prop:single-block-separation} therefore implies
\begin{equation}
    a
    \le
    \sqrt{1-\frac14e^{-10n/\sqrt{\kappa}}}
    \le
    \exp\left\{-\frac18e^{-10n/\sqrt{\kappa}}\right\},
    \label{eq:single-block-affinity}
\end{equation}
where the last inequality follows from
$\log(1-t)\le-t$ for $0\le t<1$.

Let $\bar p_b$ be the density of $\bar\pi_b$.
When $kn<d$, write $\varphi_{d-kn}$ for the standard Gaussian
density on $\mathbb R^{d-kn}$.
By Eq.~\eqref{eq:averaged-targets},
\[
    \bar p_b(x_1,\ldots,x_k,z)
    =
    \prod_{j=1}^k p_b(x_j)\,\varphi_{d-kn}(z).
\]
Tonelli's theorem and $\int\varphi_{d-kn}=1$ give
\[
\begin{aligned}
    &\int_{\mathbb R^d}
    \sqrt{\bar p_{+1}(x)\bar p_{-1}(x)}\,dx\\
    &\qquad=
    \prod_{j=1}^k
    \left(
        \int_{\mathbb R^n}
        \sqrt{p_{+1}(x_j)p_{-1}(x_j)}\,dx_j
    \right)
    \int_{\mathbb R^{d-kn}}\varphi_{d-kn}(z)\,dz\\
    &\qquad=a^k.
\end{aligned}
\]
When $kn=d$, the same identity holds with the Gaussian
factor omitted.
Finally, $\min\{u,v\}\le\sqrt{uv}$ for $u,v\ge0$ implies
\[
\begin{aligned}
    \operatorname{TV}(\bar\pi_{+1},\bar\pi_{-1})
    &=
    1-\int_{\mathbb R^d}
    \min\{\bar p_{+1}(x),\bar p_{-1}(x)\}\,dx\\
    &\ge
    1-\int_{\mathbb R^d}
    \sqrt{\bar p_{+1}(x)\bar p_{-1}(x)}\,dx\\
    &=1-a^k\\
    &\ge
    1-\exp\left\{-\frac{k}{8}e^{-10n/\sqrt{\kappa}}\right\},
\end{aligned}
\]
where the last step uses Eq.~\eqref{eq:single-block-affinity}.
\end{proof}

\subsection{Regularity}
In this section, we verify that the precision matrices in our construction belong
to $\mathcal A_{d,\kappa}$.
The spectral bounds follow from orthogonality within each block
and are preserved by the block-diagonal construction.
\label{sec:regularity}
\begin{lemma}[\bf Regularity]
For every $\boldsymbol O\in O(n)^k$,
$A_{\boldsymbol O}\in\mathcal A_{d,\kappa}$.
The potential $f_{A_{\boldsymbol O}}$ has unique minimizer $0$
and minimum value $0$.
\end{lemma}

\begin{proof}[Proof of Lemma~\ref{lem:construction-admissibility}]
For $O\in O(n)$ and $v\in\mathbb R^n$, orthogonality gives
$\|Ov\|=\|v\|$. Hence, by Cauchy--Schwarz,
$|\langle v,Ov\rangle|\le\|v\|^2$.
Since
\[
    v^\top A_Ov
    =
    \frac{\kappa+1}{2}\|v\|^2
    -
    \frac{\kappa-1}{2}\langle v,Ov\rangle,
\]
we obtain
$\|v\|^2\le v^\top A_Ov\le\kappa\|v\|^2$.
Together with $A_O=A_O^\top$, this proves
\begin{equation}
    I_n\preceq A_O\preceq\kappa I_n.
    \label{eq:single-block-regularity}
\end{equation}

Now let $x=(x_1,\ldots,x_k,z)\in\mathbb R^d$, where
$x_j\in\mathbb R^n$ and $z\in\mathbb R^{d-kn}$.
Applying Eq.~\eqref{eq:single-block-regularity} to each block yields
\[
\begin{aligned}
    \|x\|^2
    &=
    \sum_{j=1}^k\|x_j\|^2+\|z\|^2\\
    &\le
    \sum_{j=1}^k x_j^\top A_{O_j}x_j+\|z\|^2
    =
    x^\top A_{\boldsymbol O}x\\
    &\le
    \kappa\sum_{j=1}^k\|x_j\|^2+\|z\|^2
    \le
    \kappa\|x\|^2.
\end{aligned}
\]
The last inequality uses $\kappa\ge1$.
When $kn=d$, the terms involving $z$ are omitted.
Since $A_{\boldsymbol O}$ is symmetric, these inequalities imply
$I_d\preceq A_{\boldsymbol O}\preceq\kappa I_d$, and therefore
$A_{\boldsymbol O}\in\mathcal A_{d,\kappa}$.

Finally, $f_{A_{\boldsymbol O}}(0)=0$, whereas for every $x\ne0$,
\[
    f_{A_{\boldsymbol O}}(x)
    =
    \frac12x^\top A_{\boldsymbol O}x
    \ge
    \frac12\|x\|^2
    >0.
\]
Thus $0$ is the unique minimizer and the minimum value is $0$.
\end{proof}

\section{Proofs for Section~\ref{sec:5}}
\label{app:c}

\subsection{Proof of Lemma~\ref{lem:openai-rotation}}
\label{app:openai-properties}

To explain the argument, distinguish the original coordinates from the
hidden eigenbasis: $u=O^\top x$ expresses $x$ in the latter coordinates,
where the matrix acts by $u\mapsto\Lambda u$. Thus a query at $x_i$ has
hidden coordinates $u_i=O^\top x_i$ and returns
$Ax_i=O\Lambda u_i$. The algorithm observes $x_i$ and $Ax_i$ but the latent
vector $u_i$ is introduced for the analysis. The function value
$f(x_i)=\langle x_i,Ax_i\rangle/2$ is determined by this pair and adds no
information.

The difficulty is that $x_i$ may depend on earlier replies and hence on
$O$. We handle this dependence by revealing $O$ successively. Conditional
on an enlarged history, $O$ is known on a revealed latent subspace and
remains a Haar isometry between the complementary subspaces. Whenever
the unrevealed component of $O^\top x_i$ is nonzero, its direction is
uniform in the latent complement.
We convert these directions into independent standard Gaussian vectors
and propagate the resulting span relations under multiplication by
$\Lambda$. The extra vector $g_{q+1}$ accounts for the output $X$, which
is treated as one final preimage revelation.

Throughout, for subspaces $F,G\subseteq\mathbb R^d$ of equal dimension,
we use the following notation:
\begin{itemize}
\item $P_F$ is the orthogonal projection onto $F$, $\mathbb S(F)$ is the unit
  sphere of $F$, and $\operatorname{Unif}(\mathbb S(F))$ is the uniform law
  on it.
\item $\operatorname{Iso}(F,G)$ is the set of linear isometries from $F$
  onto $G$, and $\mathrm O(F)=\operatorname{Iso}(F,F)$. For
  $J\in\operatorname{Iso}(F,G)$, $J^\top\colon G\to F$ denotes the adjoint,
  so $J^\top=J^{-1}$.
\item $\operatorname{Haar}(\operatorname{Iso}(F,G))$ is the unique
  probability law on $\operatorname{Iso}(F,G)$ invariant under
  $J\mapsto UJQ^\top$ for all $U\in\mathrm O(G)$ and $Q\in\mathrm O(F)$.
  Equivalently, it is the law of $UJ_0$ for
  $U\sim\operatorname{Haar}(\mathrm O(G))$ and any fixed
  $J_0\in\operatorname{Iso}(F,G)$. If $J$ has this law, then $J^{-1}$ has
  law $\operatorname{Haar}(\operatorname{Iso}(G,F))$.
\item $\mathcal L(Y\mid\mathcal H)$ is the conditional law of a random
  element $Y$ given a $\sigma$-algebra $\mathcal H$.
\end{itemize}

The key observation is that revealing one image--preimage pair of a
Haar isometry preserves rotational symmetry on the orthogonal
complements. Once $Jv=w$ is known, orthogonal transformations fixing
$v$ and $w$ preserve both this identity and the Haar law. Consequently,
the remaining restriction is still a Haar isometry from
$F\cap v^\perp$ onto $G\cap w^\perp$.
The following fact makes this observation precise, conditionally on
the information already available.
\begin{fact}[\bf Haar revelation]
\label{fact:openai-haar}
Let $F,G$ be $\mathcal H$-measurable subspaces with $\dim F=\dim G$, and let
$\mathcal L(J\mid\mathcal H)=\operatorname{Haar}(\operatorname{Iso}(F,G))$.
For an $\mathcal H$-measurable $w\in\mathbb S(G)$, put $v=J^\top w$. Then
\begin{align*}
\mathcal L(v\mid\mathcal H)
&=\operatorname{Unif}(\mathbb S(F)),\\
\mathcal L(J|_{F\cap v^\perp}\mid\mathcal H,v)
&=\operatorname{Haar}
  \bigl(\operatorname{Iso}(F\cap v^\perp,G\cap w^\perp)\bigr).
\end{align*}
Symmetrically, for an $\mathcal H$-measurable $s\in\mathbb S(F)$ and $t=Js$,
the vector $t$ is uniform on $\mathbb S(G)$ and, given $(\mathcal H,t)$,
$J|_{F\cap s^\perp}$ is Haar on
$\operatorname{Iso}(F\cap s^\perp,G\cap t^\perp)$.
\end{fact}

\begin{proof}
Condition on $\mathcal H$. For any $Q\in\mathrm O(F)$, invariance gives
$JQ^\top\stackrel{\mathrm d}{=}J$, hence $Qv\stackrel{\mathrm d}{=}v$, so
$v$ is uniform on $\mathbb S(F)$.

Next, let $U$ be Haar on the stabilizer
$\{U\in\mathrm O(G):Uw=w\}\cong\mathrm O(G\cap w^\perp)$, independent of
$J$, and set $J'=UJ$. Then $J'\stackrel{\mathrm d}{=}J$ and
$J'^\top w=J^\top U^\top w=v$, so
$(J',v)\stackrel{\mathrm d}{=}(J,v)$. Given $J$, the restriction
$J'|_{F\cap v^\perp}=U|_{G\cap w^\perp}\circ J|_{F\cap v^\perp}$ is Haar on
$\operatorname{Iso}(F\cap v^\perp,G\cap w^\perp)$ by definition. This law
depends on $J$ only through $v$, so it is also the law given $v$, and since
$(J',v)\stackrel{\mathrm d}{=}(J,v)$ the same holds for $J$.

The symmetric statement follows by applying the above to $J^\top=J^{-1}$,
which is Haar on $\operatorname{Iso}(G,F)$.
\end{proof}

\paragraph{Haar states.}
To apply Fact~\ref{fact:openai-haar} repeatedly along an adaptive
transcript, we need to keep track of which part of $O$ has been
revealed and how the remaining part is distributed.
We record the available information by a $\sigma$-algebra
$\mathcal H$ and the revealed directions in the hidden eigenbasis
by a subspace $E$.
The history should determine both $E$ and the action of $O$ on $E$,
represented by $OP_E$, while the restriction of $O$ to $E^\perp$
should remain Haar between the corresponding orthogonal complements.
The following definition packages these requirements into a state
that we will preserve after each revelation.

\begin{definition}[\bf Haar state]
    Let $\mathcal H$ be a $\sigma$-algebra and $E\subseteq\mathbb R^d$ a random
subspace. We call $(\mathcal H,E)$ a \emph{Haar state} if $E$ and $OP_E$ are
$\mathcal H$-measurable and
\[
\mathcal L(O|_{E^\perp}\mid\mathcal H)
=\operatorname{Haar}
 \bigl(\operatorname{Iso}(E^\perp,(OE)^\perp)\bigr).
\]
\end{definition}

We next show that the Haar-state property is preserved by the
revelations needed to process an oracle query.
For a query point $x$ chosen using the current history, we first
reveal its hidden coordinates $u=O^\top x$, and then reveal the
reply $Ax=O\Lambda u$. These two steps add $u$ and $\Lambda u$,
respectively, to the revealed latent subspace.
We will also introduce independent auxiliary randomness for the
Gaussian construction below.
The following fact verifies that all three updates preserve a
Haar state and identifies the uniform direction arising from
a preimage revelation.
\begin{fact}[\bf One revelation step]
\label{fact:openai-step}
Let $(\mathcal H,E)$ be a Haar state, and write $T=OP_E$ and $W=OE$.
\begin{enumerate}
\item[(i)] \emph{(Auxiliary randomness.)}
If $Z$ is independent of $\sigma(\mathcal H,O)$, then
$(\mathcal H\vee\sigma(Z),E)$ is a Haar state.
\item[(ii)] \emph{(Revealing an image.)}
If $b\in\mathbb R^d$ is $\mathcal H$-measurable, then
$(\mathcal H\vee\sigma(Ob),\operatorname{span}(E,b))$ is a Haar state.
\item[(iii)] \emph{(Revealing a preimage.)}
If $x\in\mathbb R^d$ is $\mathcal H$-measurable, then
$(\mathcal H\vee\sigma(O^\top x),\operatorname{span}(E,O^\top x))$ is a Haar
state. Moreover, writing $x^\perp=P_{W^\perp}x$,
\begin{equation}
\label{eq:openai-preimage}
O^\top x=T^\top x+\|x^\perp\|\,v,
\qquad
v:=
\begin{cases}
O^\top x^\perp/\|x^\perp\|, & x^\perp\ne0,\\
0, & x^\perp=0,
\end{cases}
\end{equation}
and $\mathcal L(v\mid\mathcal H)=\operatorname{Unif}(\mathbb S(E^\perp))$
on the event $\{x^\perp\ne0\}$.
\end{enumerate}
\end{fact}

\begin{proof}
(i) Since $O|_{E^\perp}$ is $\sigma(\mathcal H,O)$-measurable and $Z$ is
independent of $\sigma(\mathcal H,O)$, conditioning additionally on $Z$ does
not change the conditional law of $O|_{E^\perp}$.

(iii) Since $OP_E=P_WO$, we have $T^\top x=P_EO^\top x=O^\top P_Wx$, which
gives \eqref{eq:openai-preimage}. In particular
$\mathcal H^+:=\mathcal H\vee\sigma(O^\top x)=\mathcal H\vee\sigma(v)$.
The event $\{x^\perp=0\}$ belongs to $\mathcal H$, so we may treat it and its
complement separately. On $\{x^\perp=0\}$, the vector $O^\top x=T^\top x\in E$
is $\mathcal H$-measurable and nothing changes. On $\{x^\perp\ne0\}$, the
unit vector $w=x^\perp/\|x^\perp\|\in\mathbb S(W^\perp)$ is
$\mathcal H$-measurable and $v=(O|_{E^\perp})^\top w$.
Fact~\ref{fact:openai-haar} with $J=O|_{E^\perp}$, $F=E^\perp$,
$G=W^\perp$ gives
$\mathcal L(v\mid\mathcal H)=\operatorname{Unif}(\mathbb S(E^\perp))$ and
\[
\mathcal L\bigl(O|_{E^\perp\cap v^\perp}\bigm|\mathcal H^+\bigr)
=\operatorname{Haar}
 \bigl(\operatorname{Iso}(E^\perp\cap v^\perp,W^\perp\cap w^\perp)\bigr).
\]
Let $E^+=\operatorname{span}(E,O^\top x)=E\oplus\operatorname{span}\{v\}$.
Then $(E^+)^\perp=E^\perp\cap v^\perp$,
$(OE^+)^\perp=W^\perp\cap w^\perp$, and $OP_{E^+}=T+wv^\top$ is
$\mathcal H^+$-measurable. Hence $(\mathcal H^+,E^+)$ is a Haar state.

(ii) Let $r=P_{E^\perp}b$, and set $t=Or/\|r\|$ if $r\ne0$ and $t=0$
otherwise. Then $Ob=Tb+\|r\|\,t$, so
$\mathcal H\vee\sigma(Ob)=\mathcal H\vee\sigma(t)$. The rest is as in (iii),
using the symmetric part of Fact~\ref{fact:openai-haar} with
$s=r/\|r\|\in\mathbb S(E^\perp)$ on $\{r\ne0\}$: the new revealed space is
$E\oplus\operatorname{span}\{s\}=\operatorname{span}(E,b)$, and
$OP_{\operatorname{span}(E,b)}=T+ts^\top$.
\end{proof}

\begin{fact}[\bf Gaussianizing a uniform direction]
\label{fact:openai-gaussianize}
Let $E\subseteq\mathbb R^d$ be a fixed subspace with $\dim E<d$, and let
$Z\sim N(0,I_d)$ and $v\sim\operatorname{Unif}(\mathbb S(E^\perp))$ be
independent. Then
\[
g:=P_EZ+\|P_{E^\perp}Z\|\,v\sim N(0,I_d),
\qquad
v\in E+\operatorname{span}\{g\}\quad\text{almost surely}.
\]
\end{fact}

\begin{proof}
The vectors $P_EZ\sim N(0,P_E)$ and $P_{E^\perp}Z\sim N(0,P_{E^\perp})$ are
independent. A Gaussian vector with law $N(0,P_{E^\perp})$ has independent
norm and direction, with the direction uniform on $\mathbb S(E^\perp)$.
Hence $\|P_{E^\perp}Z\|\,v\sim N(0,P_{E^\perp})$, independently of $P_EZ$,
and $g\sim N(0,P_E+P_{E^\perp})=N(0,I_d)$. Finally,
$\|P_{E^\perp}Z\|>0$ almost surely and
$\|P_{E^\perp}Z\|\,v=g-P_EZ$ with $P_EZ\in E$.
\end{proof}

\begin{proof}[Proof of Lemma~\ref{lem:openai-rotation}]
\emph{Setup.}
Collect the algorithm's private randomness in a seed $R$ independent of $O$.
Padding with dummy queries $x=0$, we may assume that the algorithm makes
exactly $q$ queries $x_1,\ldots,x_q$. A query at $x$ returns
$Ax=O\Lambda O^\top x$ and $\langle x,Ax\rangle/2$, and the latter is a
function of $(x,Ax)$. Write $x_{q+1}=X$ and $u_i=O^\top x_i$ for
$1\le i\le q+1$. Then $x_i$ is a measurable function of $R$ and the replies
$O\Lambda u_1,\ldots,O\Lambda u_{i-1}$. Extend the probability space by
i.i.d.\ vectors $Z_1,\ldots,Z_{q+1}\sim N(0,I_d)$ independent of $(R,O)$ and
this does not change the joint law of $O$, $R$, the transcript, and $X$.

\emph{Histories.}
Let $E_0=\{0\}$ and $\mathcal H_0=\sigma(R)$, and for $1\le i\le q$ let
\[
\mathcal H_i=\mathcal H_{i-1}\vee\sigma(Z_i,u_i,O\Lambda u_i),
\qquad
E_i=\operatorname{span}(E_{i-1},u_i,\Lambda u_i).
\]
Also let $\mathcal H_{q+1}=\mathcal H_q\vee\sigma(Z_{q+1},u_{q+1})$. Then
$x_{i+1}$ is $\mathcal H_i$-measurable for $0\le i\le q$. We claim that
$(\mathcal H_i,E_i)$ is a Haar state for $0\le i\le q$. For $i=0$ this holds
because $O$ is Haar and independent of $R$. For the induction step, apply
Fact~\ref{fact:openai-step} three times: part~(i) with $Z=Z_i$, part~(iii)
with $x=x_i$, and part~(ii) with $b=\Lambda u_i$. Note that
$\dim E_i\le2i\le2q<d$.

\emph{Gaussian vectors.}
Fix $1\le i\le q+1$ and write $E=E_{i-1}$. Let $x_i^\perp$ and $v_i$ be the
vectors $x^\perp$ and $v$ of \eqref{eq:openai-preimage} for the Haar state
$(\mathcal H_{i-1}\vee\sigma(Z_i),E)$ and $x=x_i$, and define
\[
g_i=
\begin{cases}
P_EZ_i+\|P_{E^\perp}Z_i\|\,v_i, & x_i^\perp\ne0,\\
Z_i, & x_i^\perp=0.
\end{cases}
\]
Then $g_i$ is $\mathcal H_i$-measurable. On $\{x_i^\perp\ne0\}$,
Fact~\ref{fact:openai-step}(iii) gives
$\mathcal L(v_i\mid\mathcal H_{i-1}\vee\sigma(Z_i))
=\operatorname{Unif}(\mathbb S(E^\perp))$.
Thus, conditionally on $\mathcal H_{i-1}$, the vectors $Z_i$ and $v_i$ are
independent with laws $N(0,I_d)$ and $\operatorname{Unif}(\mathbb S(E^\perp))$,
and Fact~\ref{fact:openai-gaussianize} gives
\begin{equation}
\label{eq:openai-gi}
\mathcal L(g_i\mid\mathcal H_{i-1})=N(0,I_d),
\qquad
u_i\in E_{i-1}+\operatorname{span}\{g_i\}\quad\text{almost surely},
\end{equation}
where the inclusion uses \eqref{eq:openai-preimage}. On
$\{x_i^\perp=0\}$, \eqref{eq:openai-gi} holds trivially since
$u_i\in E_{i-1}$.

\emph{Independence.}
By \eqref{eq:openai-gi}, the conditional law of $g_i$ given
$\mathcal H_{i-1}$ is non-random, so $g_i$ is independent of
$\mathcal H_{i-1}$. Since $g_1,\ldots,g_{i-1}$ are
$\mathcal H_{i-1}$-measurable, the vectors $g_1,\ldots,g_{q+1}$ are
mutually independent.

\emph{Krylov span.}
Let $\mathcal S_0=\{0\}$ and
$\mathcal S_i=\operatorname{span}\{\Lambda^jg_a:1\le a\le i,\ 0\le j\le i\}$.
If $E_{i-1}\subseteq\mathcal S_{i-1}$, then by \eqref{eq:openai-gi}
\[
E_i
\subseteq
\mathcal S_{i-1}+\Lambda\mathcal S_{i-1}
+\operatorname{span}\{g_i,\Lambda g_i\}
\subseteq\mathcal S_i .
\]
Hence $E_q\subseteq\mathcal S_q$ by induction, and
\[
O^\top X=u_{q+1}
\in E_q+\operatorname{span}\{g_{q+1}\}
\subseteq\mathcal S_q+\operatorname{span}\{g_{q+1}\}
\subseteq\mathcal K_q(g_1,\ldots,g_{q+1}).
\]
\end{proof}

\subsection{Proof of Lemma~\ref{lem:openai-cosine}}
\label{app:openai-cosine}
Set $r=q+1$ and $n=2q+1$. For $1\le i\le d$, define
\[
t_i=(1,\sqrt2\cos\theta_i,\ldots,
          \sqrt2\cos(q\theta_i))^\top,
\qquad
\gamma_i=((g_1)_i,\ldots,(g_r)_i)^\top.
\]
The vectors $\gamma_1,\ldots,\gamma_d$ are independent
$N(0,I_r)$. Let $B\in\mathbb R^{d\times r^2}$ have
$i$th row $d^{-1/2}(\gamma_i\otimes t_i)^\top$.

Our strategy is to write $u=Bc$ and control $\|u\|^2$
through $B^\top B$, and $u^\top D_nu$ through
$B^\top D_nB$. The following fact verifies that the
columns of $B$ span $\mathcal K_q$ and shows that these
two matrices have expectations $I_{r^2}$ and $0$,
respectively. We will then prove concentration around
these expectations to obtain a uniform bound on
$\mathcal K_q$.

\begin{fact}[\bf Cosine basis and moment identities]
\label{fact:openai-cosine-basis}
The matrix $B$ satisfies
\[
\operatorname{range}(B)=\mathcal K_q,
\qquad
\mathbb E[B^\top B]=I_{r^2},
\qquad
\mathbb E[B^\top D_nB]=0.
\]
Moreover,
\[
\|t_i\|^2\le2r,
\qquad
\sum_i t_it_i^\top=dI_r,
\qquad
\sum_i\cos(n\theta_i)t_it_i^\top=0.
\]
\end{fact}

\begin{proof}
\noindent
$\boxed{\operatorname{range}(B)=\mathcal K_q}$.
Since
$\lambda_i=(\kappa+1)/2+(\kappa-1)\cos\theta_i/2$,
the identity $\cos(j\theta)=T_j(\cos\theta)$ expresses
$\cos(j\theta_i)$ as a degree-$j$ polynomial in $\lambda_i$
with nonzero leading coefficient.
For each $a$, the change of basis between
$\bigl((g_a)_i\cos(j\theta_i)\bigr)_{i=1}^d$
and $\Lambda^jg_a$, $0\le j\le q$, is therefore
triangular and invertible.
The nonzero column normalizations in $B$ do not change
its range, proving the first identity.

\medskip
\noindent
$\boxed{\mathbb E[B^\top B]=I_{r^2}}$.
Recall that $\theta_i=(i-\tfrac12)\pi/d$.
For every integer $k$ with $0<|k|<2d$,
\begin{align*}
2\sin\!\left(\frac{k\pi}{2d}\right)
\sum_{i=1}^d\cos(k\theta_i)
&=\sum_{i=1}^d
\left[
\sin\!\left(\frac{ki\pi}{d}\right)
-\sin\!\left(\frac{k(i-1)\pi}{d}\right)
\right]\\
&=\sin(k\pi)-\sin(0)=0.
\end{align*}
Since $\sin(k\pi/(2d))\ne0$, this gives
$\sum_i\cos(k\theta_i)=0$.

Index the coordinates of $t_i$ by $0,\ldots,q$.
The entries involving the constant coordinate satisfy
\[
\sum_i(t_i)_0^2=d,
\qquad
\sum_i(t_i)_0(t_i)_j
=\sqrt2\sum_i\cos(j\theta_i)=0
\quad (1\le j\le q).
\]
For $1\le j,k\le q$, the product-to-sum identity gives
\begin{align*}
\sum_i(t_i)_j(t_i)_k
&=2\sum_i\cos(j\theta_i)\cos(k\theta_i)\\
&=\sum_i\cos((j-k)\theta_i)
 +\sum_i\cos((j+k)\theta_i)\\
&=
\begin{cases}
d,&j=k,\\
0,&j\ne k.
\end{cases}
\end{align*}
Indeed, $0<j+k\le2q<2d$, while $j-k=0$ exactly
when $j=k$; otherwise $0<|j-k|<2d$.
By symmetry, these calculations cover all matrix entries,
and hence $\sum_i t_it_i^\top=dI_r$.

Since $\mathbb E[\gamma_i\gamma_i^\top]=I_r$,
the row definition of $B$ now gives
\begin{align*}
\mathbb E[B^\top B]
&=\frac1d\sum_i
  \bigl(\mathbb E[\gamma_i\gamma_i^\top]
        \otimes t_it_i^\top\bigr)\\
&=\frac1d\left(I_r\otimes\sum_i t_it_i^\top\right)
=I_{r^2}.
\end{align*}

\medskip
\noindent
$\boxed{\mathbb E[B^\top D_nB]=0}$.
Since $n=2q+1$, for $0\le j,k\le q$ we have
\[
1\le n-j-k\le n\pm j\pm k
\le n+j+k\le4q+1<2d.
\]
Thus the cosine-sum identity above gives
\begin{align*}
&\sum_i\cos(n\theta_i)\cos(j\theta_i)\cos(k\theta_i)\\
=\quad&\frac14\sum_i\bigl[
  \cos((n+j+k)\theta_i)
 +\cos((n+j-k)\theta_i)
 +\cos((n-j+k)\theta_i)
 +\cos((n-j-k)\theta_i)
 \bigr]\\
=\quad &0.
\end{align*}
Each entry of $\sum_i\cos(n\theta_i)t_it_i^\top$
is a constant multiple of such a sum, so
$\sum_i\cos(n\theta_i)t_it_i^\top=0$.
Consequently,
\begin{align*}
\mathbb E[B^\top D_nB]
&=\frac1d\sum_i\cos(n\theta_i)
  \bigl(\mathbb E[\gamma_i\gamma_i^\top]
        \otimes t_it_i^\top\bigr)\\
&=\frac1d\left(
  I_r\otimes\sum_i\cos(n\theta_i)t_it_i^\top
\right)
=0.
\end{align*}

Finally,
\[
\|t_i\|^2
=1+2\sum_{j=1}^q\cos^2(j\theta_i)
\le1+2q\le2r.
\]
\end{proof}

We next bound the deviations of the two matrices from
their expectations in operator norm. These simultaneous
bounds will control both $\|Bc\|^2$ and
$c^\top B^\top D_nBc$ for every coefficient vector $c$.

\begin{fact}[\bf Simultaneous Gram-matrix concentration]
\label{fact:openai-gram-concentration}
For a sufficiently large absolute constant $C$ in the
definition of $\tau_{d,q}$, with probability at least
$0.995$,
\[
\|B^\top B-I_{r^2}\|_{\rm op}\le\tau_{d,q},
\qquad
\|B^\top D_nB\|_{\rm op}\le\tau_{d,q}.
\]
\end{fact}

\begin{proof}
\emph{Truncation and its bias.}
Set
$L=2r+3\log(1000d)$ and $R_i=\|\gamma_i\|^2$.
Since $\mathbb E e^{R_i/3}=3^{r/2}$,
Markov's inequality gives, for $x\ge0$,
\[
\mathbb P(R_i>2r+3x)
\le
\exp\!\left(\frac r2\log3-\frac{2r}{3}-x\right)
\le e^{-x},
\]
where the last inequality uses $(\log3)/2<2/3$.
Consequently, the event
\[
\mathcal T=\{R_i\le L\text{ for every }i\}
\]
has probability at least
$1-d/(1000d)=0.999$.

For $s\ge0$, the same tail bound gives
$\mathbb P(R_i>L+s)\le e^{-s/3}/(1000d)$.
Hence
\begin{align*}
\mathbb E[R_i\mathbf1_{\{R_i>L\}}]
&=L\mathbb P(R_i>L)
  +\int_0^\infty\mathbb P(R_i>L+s)\,ds\le\frac{L+3}{1000d}.
\end{align*}

Define the independent truncated matrices
\[
Y_i=\frac1d
(\gamma_i\gamma_i^\top\otimes t_it_i^\top)
\mathbf1_{\{R_i\le L\}}.
\]
By radial symmetry,
\[
\mathbb E[
  \gamma_i\gamma_i^\top\mathbf1_{\{R_i\le L\}}
]=aI_r,
\qquad
a=\frac1r\mathbb E[R_i\mathbf1_{\{R_i\le L\}}].
\]
The scalar $a$ is the same for every $i$.
Since $\mathbb E R_i=r$, we have $0\le a\le1$ and
\[
1-a
=\frac1r\mathbb E[R_i\mathbf1_{\{R_i>L\}}]
\le\frac{L+3}{1000rd}.
\]
Moreover,
$\mathbb E Y_i=(a/d)(I_r\otimes t_it_i^\top)$.
Fact~\ref{fact:openai-cosine-basis} therefore gives
\[
\sum_i\mathbb E Y_i=aI_{r^2},
\qquad
\sum_i\cos(n\theta_i)\mathbb E Y_i=0.
\]
Thus, the truncated weighted sum has zero expectation.

\medskip
\emph{Bounds for the centered summands.}
Put $Z_i=Y_i-\mathbb E Y_i$.
Using $\|t_i\|^2\le2r$, we obtain
\[
\|Y_i\|_{\rm op}
=\frac{R_i\|t_i\|^2}{d}
  \mathbf1_{\{R_i\le L\}}
\le\frac{2rL}{d},
\qquad
\|\mathbb E Y_i\|_{\rm op}
=\frac{a\|t_i\|^2}{d}
\le\frac{2r}{d}.
\]
Hence
\[
\|Z_i\|_{\rm op}\le
\frac{2r(L+1)}d=:R_*.
\]
The same bound holds for
$\cos(n\theta_i)Z_i$, since $|\cos(n\theta_i)|\le1$.

The Gaussian fourth moments give
\[
\mathbb E[R_i\gamma_i\gamma_i^\top]=(r+2)I_r.
\]
Indeed, the off-diagonal entries vanish by sign symmetry,
while each diagonal entry equals $3+(r-1)=r+2$.
Since
$(\gamma_i\gamma_i^\top)^2
=R_i\gamma_i\gamma_i^\top$
and
$(t_it_i^\top)^2=\|t_i\|^2t_it_i^\top$,
we have
\[
Y_i^2
=\frac{R_i\|t_i\|^2}{d^2}
(\gamma_i\gamma_i^\top\otimes t_it_i^\top)
\mathbf1_{\{R_i\le L\}}.
\]
Dropping the truncation indicator in this positive
semidefinite expression yields
\begin{align*}
\sum_i\mathbb E Y_i^2
&\preceq
\frac{r+2}{d^2}
\sum_i\|t_i\|^2(I_r\otimes t_it_i^\top)\\
&\preceq
\frac{2r(r+2)}{d^2}
\left(I_r\otimes\sum_i t_it_i^\top\right)\\
&=\frac{2r(r+2)}d I_{r^2}.
\end{align*}
Here we used $\|t_i\|^2\le2r$ and
$\sum_i t_it_i^\top=dI_r$.
Since
\[
\mathbb E Z_i^2
=\mathbb E Y_i^2-(\mathbb E Y_i)^2
\preceq\mathbb E Y_i^2,
\]
the sequence $(Z_i)_i$ has variance parameter at most
\[
v_*:=\frac{2r(r+2)}d.
\]
The same is true for $(\cos(n\theta_i)Z_i)_i$, because
\[
\mathbb E[(\cos(n\theta_i)Z_i)^2]
=\cos^2(n\theta_i)\mathbb E Z_i^2
\preceq\mathbb E Z_i^2.
\]

\medskip
\emph{Applying matrix Bernstein.}
Both sequences consist of independent, centered,
self-adjoint matrices.
The two-sided bounded matrix Bernstein inequality gives, for $t>0$,
\[
\mathbb P\!\left(
  \left\|\sum_i Z_i\right\|_{\rm op}>t
\right)
\le2r^2\exp\!\left(
  -\frac{t^2}{2(v_*+R_*t/3)}
\right),
\]
and the same bound holds for
$\sum_i\cos(n\theta_i)Z_i$.

Set
\[
z=\log(1000r^2),
\qquad
\Delta=\sqrt{2v_*z}+\frac{2R_*z}{3}.
\]
A direct expansion gives
\[
\Delta^2-2z(v_*+R_*\Delta/3)
=\frac{2R_*z}{3}\sqrt{2v_*z}\ge0.
\]
Therefore each matrix bound fails with probability
at most $2r^2e^{-z}=0.002$.
By a union bound, with probability at least $0.996$,
\[
\left\|\sum_i Z_i\right\|_{\rm op}\le\Delta,
\qquad
\left\|\sum_i\cos(n\theta_i)Z_i\right\|_{\rm op}
\le\Delta.
\]

On $\mathcal T$, the original and truncated sums agree.
Using their expectations computed above, we obtain
\[
B^\top B-I_{r^2}
=\sum_i Z_i-(1-a)I_{r^2},
\qquad
B^\top D_nB
=\sum_i\cos(n\theta_i)Z_i.
\]
Intersecting $\mathcal T$ with the two concentration
events therefore gives, with probability at least
$1-0.001-0.004=0.995$,
\[
\|B^\top B-I_{r^2}\|_{\rm op}\le\Delta+(1-a),
\qquad
\|B^\top D_nB\|_{\rm op}\le\Delta.
\]

\medskip
\emph{Simplifying the bound.}
Since $2q+1<d$, we have $d\ge2$ and $2r\le d$.
Thus $z=O(\log d)$ and $L+3=O(r+\log d)$,
with absolute implicit constants.
Together with $r+2\le3r$, these estimates imply
\begin{align*}
\Delta+(1-a)
&\le
2\sqrt{\frac{r(r+2)z}{d}}
+\frac{4r(L+1)z}{3d}
+\frac{L+3}{1000rd}\\
&\le
c\left[
r\sqrt{\frac{\log d}{d}}
+\frac{r(r+\log d)\log d}{d}
\right]
\end{align*}
for an absolute constant $c>0$.
Since $r=q+1$, choosing the absolute constant
$C$ in $\tau_{d,q}$ so that $C\ge c$ gives
$\Delta+(1-a)\le\tau_{d,q}$.
This proves both asserted bounds.
\end{proof}

\begin{proof}[Proof of Lemma~\ref{lem:openai-cosine}]
Choose the absolute constant $C$ sufficiently large
for Fact~\ref{fact:openai-gram-concentration} to apply.
Then, with probability at least $0.995$,
\[
\|B^\top B-I_{r^2}\|_{\rm op}\le\tau_{d,q},
\qquad
\|B^\top D_nB\|_{\rm op}\le\tau_{d,q}.
\]
On this event, let $u\in\mathcal K_q$.
By Fact~\ref{fact:openai-cosine-basis}, we can write
$u=Bc$. The two matrix bounds give
\begin{align*}
\|u\|^2
&=c^\top B^\top Bc
\ge(1-\tau_{d,q})\|c\|^2,\\
|u^\top D_nu|
&=|c^\top B^\top D_nBc|
\le\tau_{d,q}\|c\|^2.
\end{align*}
Since $\tau_{d,q}\le1/2$, combining these inequalities
yields
\[
|u^\top D_nu|
\le\frac{\tau_{d,q}}{1-\tau_{d,q}}\|u\|^2
\le2\tau_{d,q}\|u\|^2.
\]
These inequalities hold for every $c$ on the same event,
so the conclusion holds simultaneously for every
$u\in\mathcal K_q$.
Since $0.995\ge0.99$, the lemma follows.
\end{proof}

\subsection{Proof of Lemma~\ref{lem:openai-target-bias}}
\label{app:openai-target-bias}

Recall that
$\theta_i=(i-\tfrac12)\pi/d$ and
$\lambda_i=(\kappa+1)/2+(\kappa-1)\cos\theta_i/2$.
We first expand the inverse spectrum in cosine modes
and sum this expansion over the grid. We then bound
the variances and the normalized cosine bias.

\begin{fact}[\bf Cosine expansion of the inverse spectrum]
\label{fact:openai-inverse-cosine}
For every $\theta\in\mathbb R$,
\[
\frac{2}{\kappa+1+(\kappa-1)\cos\theta}
=\frac1{\sqrt\kappa}
\left(
1+2\sum_{j\ge1}(-\varrho_\kappa)^j\cos(j\theta)
\right).
\]
The series converges absolutely and uniformly in $\theta$.
\end{fact}

\begin{proof}
Since $0<\varrho_\kappa<1$, the geometric-series
identity gives
\begin{align*}
1+2\sum_{j\ge1}(-\varrho_\kappa)^j\cos(j\theta)
&=1-\frac{\varrho_\kappa e^{i\theta}}
          {1+\varrho_\kappa e^{i\theta}}
     -\frac{\varrho_\kappa e^{-i\theta}}
          {1+\varrho_\kappa e^{-i\theta}}\\
&=\frac{1-\varrho_\kappa^2}
        {1+2\varrho_\kappa\cos\theta+\varrho_\kappa^2}.
\end{align*}
The definition of $\varrho_\kappa$ implies
\[
\frac{1-\varrho_\kappa}{1+\varrho_\kappa}
=\frac1{\sqrt\kappa},
\qquad
\frac{1+2\varrho_\kappa\cos\theta+\varrho_\kappa^2}
     {(1-\varrho_\kappa)^2}
=\frac{\kappa+1+(\kappa-1)\cos\theta}{2}.
\]
Combining these identities proves the expansion.
Finally,
\[
\sum_{j\ge1}
\left|(-\varrho_\kappa)^j\cos(j\theta)\right|
\le\sum_{j\ge1}\varrho_\kappa^j
=\frac{\varrho_\kappa}{1-\varrho_\kappa},
\]
which proves absolute and uniform convergence.
\end{proof}

\begin{fact}[\bf Exact finite-grid moments]
\label{fact:openai-target-moments}
For $1\le n<d$,
\[
M=\frac d{\sqrt\kappa}
  \frac{1-\varrho_\kappa^{2d}}
       {1+\varrho_\kappa^{2d}},
\qquad
(-1)^n\mu_n
=\frac d{\sqrt\kappa}
  \frac{\varrho_\kappa^n-\varrho_\kappa^{2d-n}}
       {1+\varrho_\kappa^{2d}}.
\]
\end{fact}

\begin{proof}
For every integer $k$,
\[
\sum_{i=1}^d\cos(k\theta_i)
=
\begin{cases}
d(-1)^a,&k=2ad,\quad a\in\mathbb Z,\\
0,&k\notin2d\mathbb Z.
\end{cases}
\]
Indeed, if $k\notin2d\mathbb Z$, the finite cosine sum is
$\sin(k\pi)/(2\sin(k\pi/(2d)))=0$.
If $k=2ad$, then
$\cos(k\theta_i)=\cos(2ai\pi-a\pi)=(-1)^a$.

Since $\mathbb E U_i^2=1/\lambda_i$,
\[
M=\sum_i\frac1{\lambda_i},
\qquad
\mu_n=\sum_i\frac{\cos(n\theta_i)}{\lambda_i}.
\]
Fact~\ref{fact:openai-inverse-cosine} allows us to
interchange the absolutely convergent series with
the finite grid sum.

For $M$, only the terms with $j=2ad$, $a\ge1$,
survive. Thus
\begin{align*}
M
&=\frac d{\sqrt\kappa}
  \left(1+2\sum_{a\ge1}
                (-1)^a\varrho_\kappa^{2ad}\right)\\
&=\frac d{\sqrt\kappa}
  \left(1-\frac{2\varrho_\kappa^{2d}}
                  {1+\varrho_\kappa^{2d}}\right)\\
&=\frac d{\sqrt\kappa}
  \frac{1-\varrho_\kappa^{2d}}
       {1+\varrho_\kappa^{2d}}.
\end{align*}

For $\mu_n$, we have $\sum_i\cos(n\theta_i)=0$.
Using the product-to-sum identity gives
\[
\mu_n
=\frac1{\sqrt\kappa}\sum_{j\ge1}(-\varrho_\kappa)^j
\left[
\sum_i\cos((j-n)\theta_i)
+\sum_i\cos((j+n)\theta_i)
\right].
\]
The first inner sum is nonzero precisely when $j=n$
or $j=2ad+n$ with $a\ge1$; its value is respectively
$d$ or $d(-1)^a$.
The second inner sum is nonzero precisely when
$j=2ad-n$ with $a\ge1$, and then equals $d(-1)^a$.
These three classes are disjoint because $1\le n<d$.
Moreover,
$(-\varrho_\kappa)^{2ad\pm n}
=(-1)^n\varrho_\kappa^{2ad\pm n}$.
Consequently,
\begin{align*}
(-1)^n\mu_n
&=\frac d{\sqrt\kappa}
\left[
\varrho_\kappa^n+
\sum_{a\ge1}(-1)^a
\left(
\varrho_\kappa^{2ad+n}
+\varrho_\kappa^{2ad-n}
\right)
\right]\\
&=\frac d{\sqrt\kappa}
\left[
\varrho_\kappa^n
-\frac{\varrho_\kappa^{2d+n}
       +\varrho_\kappa^{2d-n}}
      {1+\varrho_\kappa^{2d}}
\right]\\
&=\frac d{\sqrt\kappa}
  \frac{\varrho_\kappa^n-\varrho_\kappa^{2d-n}}
       {1+\varrho_\kappa^{2d}}.
\end{align*}
\end{proof}

\begin{proof}[Proof of Lemma~\ref{lem:openai-target-bias}]
Fact~\ref{fact:openai-target-moments} gives the
two asserted means.

\emph{Variance bounds.}
Since $\Lambda$ is diagonal, the coordinates of $U$
are independent. Each $U_i$ is centered Gaussian with
variance $1/\lambda_i$, so
$\mathbb E U_i^4=3/\lambda_i^2$ and
$\operatorname{Var}(U_i^2)=2/\lambda_i^2$.
Also,
$\lambda_i\ge(\kappa+1)/2-(\kappa-1)/2=1$.
Hence
\begin{align*}
\operatorname{Var}(\|U\|^2)
&=2\sum_i\frac1{\lambda_i^2}
\le2\sum_i\frac1{\lambda_i}=2M,\\
\operatorname{Var}(U^\top D_nU)
&=2\sum_i\frac{\cos^2(n\theta_i)}{\lambda_i^2}
\le2M.
\end{align*}
The exact formula for $M$ gives $M\le d/\sqrt\kappa$.

\medskip
\emph{Normalized bias.}
Since $n<d$ and $0<\varrho_\kappa<1$,
we have
$\varrho_\kappa^n>\varrho_\kappa^{2d-n}$.
Thus $(-1)^n\mu_n>0$, and dividing the two moment
formulas yields
\[
\frac{|\mu_n|}{M}
=\varrho_\kappa^n
  \frac{1-\varrho_\kappa^{2(d-n)}}
       {1-\varrho_\kappa^{2d}}.
\]
Because
$0<\varrho_\kappa^{2d}
<\varrho_\kappa^{2(d-n)}<1$,
the fraction satisfies
\[
1-\varrho_\kappa^{2(d-n)}
\le
\frac{1-\varrho_\kappa^{2(d-n)}}
     {1-\varrho_\kappa^{2d}}
\le1.
\]

To bound the powers of $\varrho_\kappa$, observe that
\begin{align*}
\log\frac1{\varrho_\kappa}
&=\log\frac{\sqrt\kappa+1}{\sqrt\kappa-1}\\
&=2\int_0^{1/\sqrt\kappa}\frac{dt}{1-t^2}\\
&=\frac2{\sqrt\kappa}
  +2\int_0^{1/\sqrt\kappa}\frac{t^2}{1-t^2}\,dt.
\end{align*}
The last integral is nonnegative, giving
$\log(1/\varrho_\kappa)\ge2/\sqrt\kappa$.
Throughout the integration interval,
$1-t^2\ge1-\kappa^{-1}$, so
\begin{align*}
\log\frac1{\varrho_\kappa}
&\le\frac2{\sqrt\kappa}
  +\frac{2}{1-\kappa^{-1}}
     \int_0^{1/\sqrt\kappa}t^2\,dt\\
&=\frac2{\sqrt\kappa}
  \left(1+\frac1{3(\kappa-1)}\right)
\le\frac8{3\sqrt\kappa}
<\frac3{\sqrt\kappa}.
\end{align*}
Here we used $\kappa\ge2$.
Exponentiating these bounds gives
\[
e^{-3n/\sqrt\kappa}
\le\varrho_\kappa^n
\le e^{-2n/\sqrt\kappa}.
\]

If $d-n\ge\sqrt\kappa$, the lower bound on
$\log(1/\varrho_\kappa)$ also gives
\[
\varrho_\kappa^{2(d-n)}
=\exp\!\left(
-2(d-n)\log\frac1{\varrho_\kappa}
\right)
\le e^{-4(d-n)/\sqrt\kappa}
\le e^{-4}.
\]
Consequently,
$(1-e^{-4})\varrho_\kappa^n
\le|\mu_n|/M\le\varrho_\kappa^n$.
Combining this with the bounds on $\varrho_\kappa^n$
proves
\[
(1-e^{-4})e^{-3n/\sqrt\kappa}
\le\frac{|\mu_n|}{M}
\le e^{-2n/\sqrt\kappa}.
\]
\end{proof}

\subsection{Proof of Lemma~\ref{lem:openai-admissibility}}
\label{app:openai-admissibility}

\begin{proof}[Proof of Lemma~\ref{lem:openai-admissibility}]
Since $0<\theta_i<\pi$, we have $-1<\cos\theta_i<1$,
and hence $1<\lambda_i<\kappa$.
As $A=O\Lambda O^\top$ with $O$ orthogonal, this gives
$I_d\preceq A\preceq\kappa I_d$.

The potential $V_A(x)=\tfrac12x^\top Ax$ is $C^\infty$
with Hessian $\nabla^2V_A(x)=A$.
Thus it is $1$-strongly convex and $\kappa$-smooth.
Moreover, $V_A(0)=0$ and
$V_A(x)\ge\tfrac12\|x\|^2>0$ for $x\ne0$,
so $0$ is its unique minimizer.

The associated target is $N(0,A^{-1})$.
Gaussian integration gives
\[
Z_A
=(2\pi)^{d/2}(\det A)^{-1/2}
=(2\pi)^{d/2}(\det\Lambda)^{-1/2}.
\]
Since $\Lambda$ is fixed, $Z_A$ is independent of $O$
and is known whenever the spectrum is public.
\end{proof}
\end{document}